\documentclass[10pt]{article}
\usepackage{epsf}
\usepackage{amsmath}

\allowdisplaybreaks

\usepackage[showframe=false]{geometry}
\usepackage{changepage}

\usepackage{epsfig}
\usepackage{amssymb}

\usepackage{amsthm}
\usepackage{setspace}
\usepackage{cite}
\usepackage{mcite}

\usepackage{algorithmic}  
\usepackage{algorithm}

\usepackage{shadow}
\usepackage{fancybox}
\usepackage{fancyhdr}

\usepackage{color}
\usepackage[usenames,dvipsnames,svgnames,table]{xcolor}

\definecolor{mypurple}{rgb}{.4,.0,.5}

\usepackage{xcolor}

\usepackage{color}

\definecolor{darkgreen}{rgb}{0, 0.4,0}
\newcommand{\dgr}[1]{\textcolor{darkgreen}{#1}}

\definecolor{purplebrown}{rgb}{0.5,0.1,0.6}

\newcommand{\bl}[1]{\textcolor{blue}{#1}}

\definecolor{shadebrown}{rgb}{0.1,0.1,0.9}
\definecolor{lightblue}{rgb}{0.2,0,1}

\usepackage{fancybox}
\usepackage{graphicx}
\usepackage{epstopdf}
\usepackage{epsfig}
\usepackage{wrapfig}
\usepackage{subfigure}

\usepackage{xcolor}

\usepackage{tcolorbox}
\tcbuselibrary{skins}

\newtcbox{\xmybox}{on line,
arc=7pt,
before upper={\rule[-3pt]{0pt}{10pt}},boxrule=0pt,
boxsep=0pt,left=6pt,right=6pt,top=0pt,bottom=0pt,enhanced, coltext=blue, colback=white!10!yellow}

\newtcbox{\xmyboxa}{on line,
arc=7pt,
before upper={\rule[-3pt]{0pt}{10pt}},boxrule=0pt,
boxsep=0pt,left=6pt,right=6pt,top=0pt,bottom=0pt,enhanced, colback=white!10!yellow}

\newtcbox{\xmyboxb}{on line,
arc=7pt,
before upper={\rule[-3pt]{0pt}{10pt}},boxrule=1pt,colframe=darkgreen!100!blue,
boxsep=0pt,left=6pt,right=6pt,top=0pt,bottom=0pt,enhanced, colback=white!10!yellow}

\newtcbox{\xmyboxc}{on line,
arc=7pt,
before upper={\rule[-3pt]{0pt}{10pt}},boxrule=.7pt,colframe=blue!100!blue,
boxsep=0pt,left=6pt,right=6pt,top=0pt,bottom=0pt,enhanced, coltext=blue, colback=white!10!yellow}

\newtcbox{\xmytboxa}{on line,
arc=7pt,
before upper={\rule[-3pt]{0pt}{10pt}},boxrule=.0pt,colframe=pink!50!yellow,
boxsep=0pt,left=6pt,right=6pt,top=0pt,bottom=0pt,enhanced, coltext=white, colback=blue!40!red}

\newtcbox{\xmytboxb}{on line,
arc=7pt,
before upper={\rule[-3pt]{0pt}{10pt}},boxrule=.0pt,colframe=pink!50!yellow,
boxsep=0pt,left=6pt,right=6pt,top=0pt,bottom=0pt,enhanced, coltext=white, colback=white!40!green}

\usepackage[hyphens]{url}

\usepackage[colorlinks=true,
            linkcolor=black,
            urlcolor=blue,
            citecolor=purple]{hyperref}

\usepackage{breakurl}

\def\y{{\bf y}}
\def\v{{\bf v}}
\def\x{{\bf x}}

\def\x{{\mathbf x}}

\def\v{{\bf v}}
\def\x{{\bf x}}
\def\y{{\bf y}}
\def\z{{\bf z}}

\def\h{{\bf h}}

\def\be{\begin{equation}}
\def\ee{\end{equation}}
\def\ba{\left[\begin{array}}
\def\ea{\end{array}\right]}

\def\v{{\bf v}}
\def\x{{\bf x}}
\def\y{{\bf y}}
\def\z{{\bf z}}

\def\1{{\bf 1}}

\def\g{{\bf g}}
\def\0{{\bf 0}}

\def\erfinv{\mbox{erfinv}}
\def\erf{\mbox{erf}}
\def\erfc{\mbox{erfc}}

\def\mR{{\mathbb R}}

\def\mE{{\mathbb E}}

\newtheorem{theorem}{Theorem}

\begin{document}

\begin{singlespace}

\title {Controlled Loosening-up (CLuP) -- achieving \emph{exact} MIMO ML in polynomial time 
}
\author{
\textsc{Mihailo Stojnic
\footnote{e-mail: {\tt flatoyer@gmail.com}} }}
\date{}
\maketitle

\centerline{{\bf Abstract}} \vspace*{0.1in}

In this paper we attack one of the most fundamental signal processing/informaton theory problems, widely known as the MIMO ML-detection. We introduce a powerful Random Duality Theory (RDT) mechanism that we refer to as the Controlled Loosening-up (CLuP) as a way of achieving the \textbf{\emph{exact}} ML-performance in MIMO systems in polynomial time. We first outline the general strategy and then discuss the rationale behind the entire concept. A solid collection of results obtained through numerical experiments is presented as well and found to be in an excellent agreement with what the theory predicts. As this is the introductory paper of a massively general concept that we have developed, we mainly focus on keeping things as simple as possible and put the emphasis on the most fundamental ideas. In our several companion papers we present various other complementary results that relate to both, theoretical and practical aspects and their connections to a large collection of other problems and results that we have achieved over the years in Random Duality.

\vspace*{0.25in} \noindent {\bf Index Terms: ML - detection; MIMO systems; Algorthms; Random duality theory}.

\end{singlespace}

\section{Introduction}
\label{sec:back}

The MIMO ML-detection is one of the most fundamental open problem at the intersection of a variety of scientific fields, most notably, the information theory, signal processing, statistics, and algorithmic optimization. Due to its enormous popularity it basically needs no introduction and we will consequently try to skip as much of unnecessary repetitive introductive detailing as possible and focus only on the key points. To that end we start with a MIMO linear system which is typically modeled in the following way:
\begin{eqnarray}\label{eq:linsys1}
\y=A\x_{sol}+\sigma\v.
\end{eqnarray}
In (\ref{eq:linsys1}) $\x_{sol}\in\mR^n$ is the input vector of the system, $A\in\mR^{m\times n}$ is the system matrix, $\v\in\mR^m$ is the noise vector at the output of the system scaled by a factor $\sigma$, and $\y\in\mR^m$ is the output vector of the system. For example, in information theory a multi-antenna system is typically modelled through (\ref{eq:linsys1}). In such a system $n$ is the number of the transmitting antennas, $m$ is the number of the receiving antennas, $\x_{sol}$ is the transmitted signal vector, $A$ is the channel matrix, $\v$ is the noise vector at the receiving antennas, and $\y$ is the vector that is finally received. In this paper we will focus on a statistical and large dimensional MIMO setup which is also very typical in various applications in communications, control, and information theory. Namely, we will assume that the elements of both, $A$ and $\v$, are i.i.d. standard normals and that both, $n$ and $m$, are large so that $m=\alpha n$ where $\alpha>0$ is a real number. Moreover, we will consider the so-called \emph{coherent} scenario where the matrix $A$ is known at the receiving end and one wonders how the transmitted signal $\x$ can be estimated given $\y$ and $A$. In such a scenario one then typically relies on the so-called ML estimate which, due to the Gaussianity of $\v$, effectively boils down to solving the following problem
\begin{eqnarray}\label{eq:ml1}
\hat{\x}=\min_{\x\in{\cal X}}\|\y-A\x\|_2,
\end{eqnarray}
where ${\cal X}$ stands for the set of permissible $\x$. For the simplicity of the exposition we will assume the standard binary scenarios, i.e. ${\cal X}=\{-\frac{1}{\sqrt{n}},\frac{1}{\sqrt{n}}\}^n$. However, we do mention that the mechanisms that we present below can easily be adapted to fit various other scenarios as well.

The optimization in (\ref{eq:ml1}) is of course well known and belongs to the class of the least-squares problems often seen in many fields ranging from say statistics and machine learning to information theory, communications, and control. Over last several decades in many of these fields various different techniques have been developed to attack these problems. The level of difficulty of these problems is different from one field to another and it depends to a large degree on the structure of set ${\cal X}$. The experienced reader will immediately recognize that the above assumed structure of ${\cal X}$ makes the problem that we will consider in this paper notoriously hard and quite likely among the hardest widely popular simple to state algorithmic problems. The key difficulty of course comes from the fact that the set ${\cal X}$ is discrete and known continuous optimization techniques that run  in an acceptable computational time (say polynomial) typically fail to solve such problems \textbf{\emph{exactly}}. Still, even this particular version of the problem has been the subject of an extensive research over last several decades and there has been a lot of great work that was done to improve its general understanding. We leave the details of all the prior work to survey papers and here focus only on a couple of papers that are most directly related.

As an alternative to continuous heuristics that typically solve the problem only approximatively (therefore inducing an additional residual error in the estimated $\hat{\x}$), in our own line of work initiated in \cite{StojnicBBSD08,StojnicBBSD05} we approached the problem looking for the \textbf{\emph{exact}} solutions. We designed a branch-and-bound procedure that substantially improved over the state of the art so-called Sphere-decoder (SD) algorithm of \cite{FinPhoSD85,HassVik05,JalOtt05}. As a tree-search algorithm, it had as its best feature the ability to prune the search tree way more significantly than the original SD. That of course substantially dropped the computational complexity and brought it to be close to polynomial in a wide range of systems parameters. Still, breaking the exponential/polynomial barrier remained as an unreachable goal. This barrier is precisely what we attack below. However, as it will soon become clear, some of the ideas will have certain connections to the roots of the main ideas that we introduced in  \cite{StojnicBBSD08,StojnicBBSD05} but the key components are actually completely different and will in fact mostly rely on the very powerful concept called \bl{\textbf{Random Duality Theory}} (RDT) that we designed for handling a large class of optimization problems, among many of them the well-known LASSO/SOCP variants of (\ref{eq:ml1}) (see, e.g. \cite{StojnicGenLasso10,StojnicGenSocp10,StojnicPrDepSocp10}), typically seen in various settings in statistics, compressed sensing, and machine learning (see also, e.g. \cite{CheDon95,Tibsh96,DonMalMon10,BunTsyWeg07,vandeGeer08,MeinYu09}).

The presentation below will be split into several main parts. We will first introduce the main algorithm that will be utilized for solving (\ref{eq:ml1}). In the second part we will discuss its performance and the rationale behind the algorithm's structure. Finally, in the third part we will provide a substantial set of numerical results, both theoretical and practical, that will demonstrate the full power of the introduced concepts.

\section{Controlled Loosening-up (CLuP)}
\label{sec:clup}

Let $\x^{(0)}$ be a randomly generated vector from ${\cal X}=\{-\frac{1}{\sqrt{n}},\frac{1}{\sqrt{n}}\}^n$. We consider the following iterative procedure to solve (\ref{eq:ml1}):
\begin{eqnarray}
\x^{(i+1)}=\frac{\x^{(i+1,s)}}{\|\x^{(i+1,s)}\|_2} \quad \mbox{with}\quad \x^{(i+1,s)}=\mbox{arg}\min_{\x} & & -(\x^{(i)})^T\x  \nonumber \\
\mbox{subject to} & & \|\y-A\x\|_2\leq r\nonumber \\
&& \x\in \left [-\frac{1}{\sqrt{n}},\frac{1}{\sqrt{n}}\right ]^n, \label{eq:clup1}
\end{eqnarray}
where $r$ is a carefully chosen \emph{radius}. We will refer to the above procedure as Controlled Loosening-up (CLuP). The procedure looks incredibly simple and one immediately wonders why it would have any chance to be successful. The general answer is very complicated but here we will just briefly hint at why it actually may be a good idea to use the above procedure. First, in the constraint set of the inner optimization one recognizes a problem that in a way resembles the so-called \textbf{\emph{polytope}} relaxation that we actually introduced as a first step (and later a lower bounding technique) in the branch-and-bound mechanism in \cite{StojnicBBSD05,StojnicBBSD08}. One should of course immediately note a couple of important points. First, it is not the polytope relaxation itself but rather a specifically constrained problem that has the discrete $n$-cube vertices set relaxed to a polytope (basically a full $n$-cube). Second, when we introduced the branch-and-bound mechanism in \cite{StojnicBBSD05,StojnicBBSD08} we immediately recognized that the polytope relaxation is a nice heuristic but on its own essentially hopeless when it comes to finding the \textbf{\emph{exact}} solution of (\ref{eq:ml1}). Here though the idea is completely different. The discrete set is relaxed to a convex continuous one so that the optimization in (\ref{eq:clup1}) can be solved quickly in polynomial time. The key point is in carefully choosing $r$ and hoping that such a careful choice may eventually lead to an ML solution.

\begin{algorithm}[t]
\caption{Controlled Loosening-up (CLuP -- achieving \textbf{exact} ML in polynomial time)}

{\bf Input:} Received vector $\y \in \mR^m$, system matrix $A\in \mR^{m\times n}$, radius $r$, starting unknown vector $\x^{(0)}\in\mR^n$, set of additional (convex) constraints ${\cal A}(\x)$ (empty set is fine as well), maximum number of iterations $i_{max}$, desired converging precision $\delta_{min}$.[\bl{$\mbox{CLuP}(\y,A,r,\x^{(0)},{\cal A}(\x),i_{max},\delta)$}] \\
{\bf Output:} Estimated vector $\x^{(i)} \in \mR^n$ and its discretized  variant $\x^{(CLuP)}$.[\bl{$\x^{(i)},\x^{(CLuP)}$}]

\begin{algorithmic}[1]

\STATE Initialize the convergence gap and the iteration counter, $\delta\leftarrow 10^{10}$ and $i\leftarrow 0$

\STATE Set $c_2^{(0)}\leftarrow \delta^2$

\WHILE{$i+1\leq i_{max}$ and/or $\delta\geq\delta_{min}$}

\STATE Obtain $\x^{(i+1,s)}$ as the optimal solution of the following convex optimization problem
\begin{eqnarray}
\x^{(i+1,s)}=\mbox{arg}\min_{\x} & & -(\x^{(i)})^T\x  \nonumber \\
\mbox{subject to} & & \|\y-A\x\|\leq r\nonumber \\
&& \x\in \left [-\frac{1}{\sqrt{n}},\frac{1}{\sqrt{n}}\right ]^n\nonumber \\
& & {\cal A}(\x).\nonumber
\end{eqnarray}

\STATE Set
\begin{eqnarray}
\x^{(i+1)}=\frac{\x^{(i+1,s)}}{\|\x^{(i+1,s)}\|_2}.  \nonumber
\end{eqnarray}

\STATE Set $c_2^{(i+1)}\leftarrow (-(\x^{(i)})^T\x^{(i+1,s)})^2$

\STATE Set $\delta\leftarrow |\sqrt{c_2^{(i+1)}}-\sqrt{c_2^{(i)}}|$

\STATE Update the iteration counter $i\leftarrow i+1$

\ENDWHILE

\STATE $\x^{(CLuP)}\leftarrow \frac{1}{\sqrt{n}}\mbox{sign}(\x^{(i)})$.

\end{algorithmic}

\end{algorithm}

Before, moving further with the discussion related to the choice of $r$ we in Figure \ref{fig:fighighlightclup} highlight the performance of the CLuP algorithm introduced above. We chose, $\alpha=0.8$. The experienced reader will already here recognize that with this choice we are already getting into the regimes where the MIMO ML-detection problem starts to become very difficult and where the know techniques might start having problems trying to reach not only the exact solution but even a good approximate one. It is probably needless to say that as $\alpha$ decreases the problem becomes harder and harder and for $\alpha\rightarrow 0$ approaches one of the hardest well-known optimization problems where hardly any solving technique is known to be of much use.

In addition to the plot that corresponds to the CLuP's probability of error we also showed the probability of errors of typical convex relaxation based heuristics, as well as the estimate for the ML. We chose the three probably most popular convex relaxation heuristics, the Ball-relaxtion, the Polytope-relaxation, and the SDP-relaxation. These are, of course, well known techniques in the optimization theory (see, e.g. \cite{GolVanLoan96Book,GroLovSch93Book,vanMaarWar00,GoeWill95}) and we considered them as the starting points and later on as the lower-bounding techniques of the branch-and-bound mechanism that we designed in \cite{StojnicBBSD05,StojnicBBSD08} for attacking on the so-called \textbf{exact} level this very same MIMO ML-detection problem. Although it is very well known we recall that: 1) the Ball heuristic relaxes ${\cal X}$ to the unit $n$-dimensional ball, 2) the Polytpe heuristic relaxes ${\cal X}$ to the unit cube, and 3) the SDP heuristic relaxes $\begin{bmatrix}\x \\1/\sqrt{n}\end{bmatrix}\begin{bmatrix}\x \\1/\sqrt{n}\end{bmatrix}^T$ to a full rank $n$-scaled unit diagonal positive semi-definite matrix. Of course, there are many other more sophisticated relaxations that one can quickly design. As this paper does not have heuristic type of approach as its main topic we selected the above three as historically and conceptually probably the most relevant ones.
\begin{figure}[htb]
\centering
\centerline{\epsfig{figure=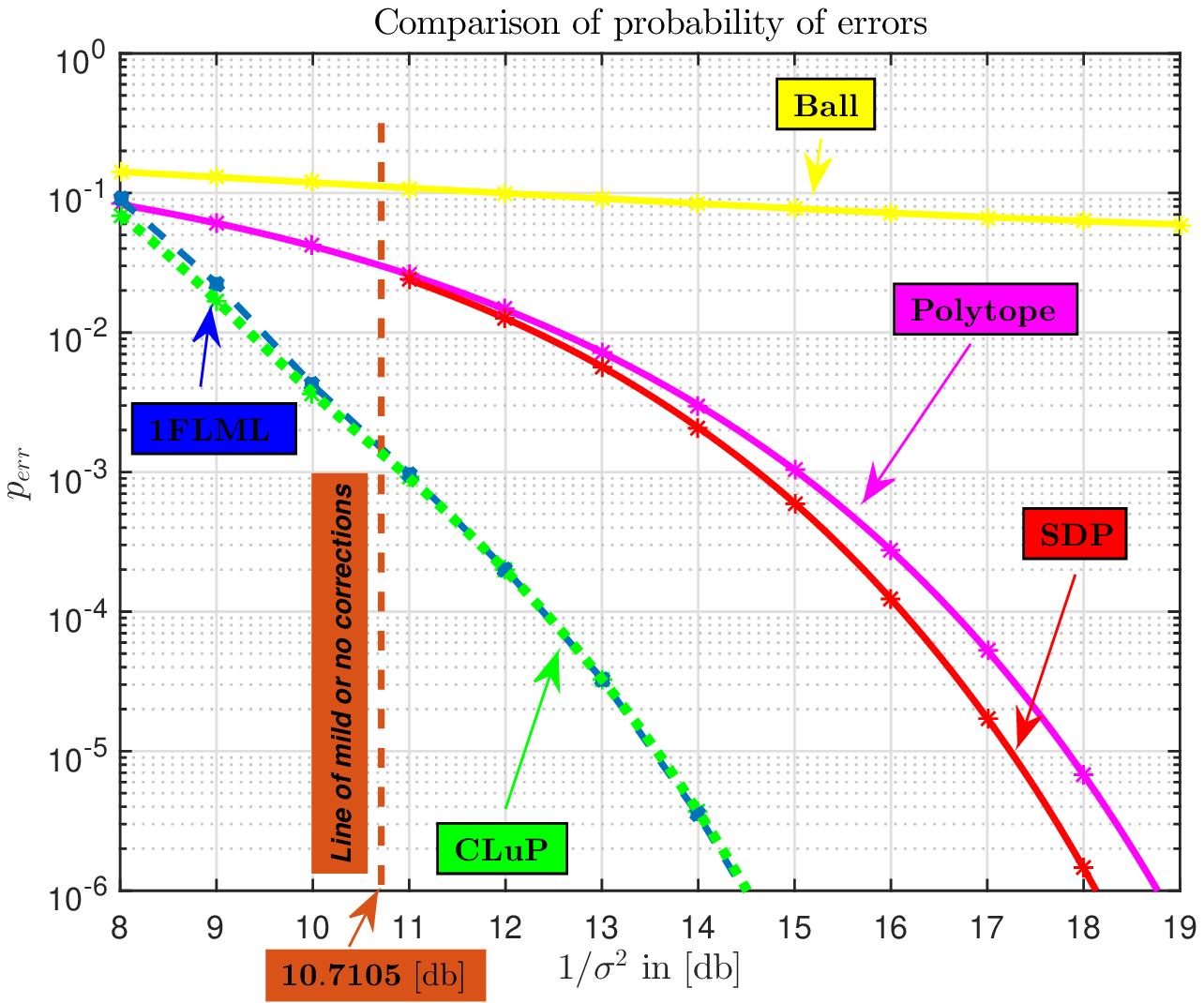,width=11.5cm,height=8cm}}
\caption{Comparison of $p_{err}$ as a function of $1/\sigma^2$; $\alpha=0.8$}
\label{fig:fighighlightclup}
\end{figure}
With the appearance of the \bl{\textbf{Random Duality Theory}} (RDT) calculating the performance characterizations of all these convexity based techniques is relatively simple and the resulting plots are given in Figure \ref{fig:fighighlightclup} (estimating the ML performance though is a bit more complicated task and we will discuss some of its intricacies below). From Figure \ref{fig:fighighlightclup} one can expect CLuP to substantially outperform the convexity based techniques. The appearance of the so-called vertical line of corrections already right here at the beginning indicates that things are not as simple as the algorithm's structure and these plots make them to be. It is of course impossible to understand the meaning of this line right here. We just mention in passing that we will have a whole lot more to say about it later on. For the time being though, one can simply think of the SNR regimes above the line as the ones of main interest (where the probabilities of error start to rapidly go down) and where things are likely to be indeed as simple as the structure of the algorithm and the plots make them to be.

\subsection{Choosing $r$}
\label{sec:choosingr}

It is rather simple to see that the above CLuP procedure will converge. To simplify writing we will assume that the converging solution is $\x$ and look at the structure of the resulting ending optimization
\begin{eqnarray}
\min_{\x} & & -\|\x\|_2  \nonumber \\
\mbox{subject to} & & \|\y-A\x\|_2\leq r\nonumber \\
&& \x\in \left [-\frac{1}{\sqrt{n}},\frac{1}{\sqrt{n}}\right ]^n. \label{eq:clup2}
\end{eqnarray}
To characterize the performance of the above optimization we of course rely on the Random Duality Theory (RDT) that we have developed in a long line of work \cite{StojnicCSetam09,StojnicCSetamBlock09,StojnicISIT2010binary,StojnicDiscPercp13,StojnicUpper10,StojnicGenLasso10,StojnicGenSocp10,StojnicPrDepSocp10,StojnicRegRndDlt10,Stojnicbinary16fin,Stojnicbinary16asym}. Before formally redoing the RDT steps we note that (\ref{eq:clup2}) is structurally the same problem as the one in \cite{StojnicGenLasso10,StojnicGenSocp10,StojnicPrDepSocp10,Stojnicbinary16fin,Stojnicbinary16asym} with a tiny change in the set of constraints. Moreover, the same set of constraints we have already considered in \cite{StojnicISIT2010binary,StojnicDiscPercp13}. As was the case in \cite{StojnicGenLasso10,StojnicGenSocp10,StojnicPrDepSocp10,Stojnicbinary16fin,Stojnicbinary16asym,StojnicISIT2010binary,StojnicDiscPercp13} we will again without a loss of generality assume that $\x_{sol}$ has a particular structure. Here, we will say that its all components are equal to $\frac{1}{\sqrt{n}}$. We will also set,
\begin{eqnarray}
c_2 & = & \|\x\|_2^2\nonumber \\
c_1 & = & (\x_{sol})^T\x, \label{eq:clup3}
\end{eqnarray}
and rewrite (\ref{eq:clup2}) in the following way
\begin{eqnarray}
\min_{\x} & & -\|\x\|_2  \nonumber \\
\mbox{subject to} & & \|[A \v]\begin{bmatrix}\x_{sol}-\x\\\sigma\end{bmatrix}\|_2\leq r\nonumber \\
&& \x\in \left [-\frac{1}{\sqrt{n}},\frac{1}{\sqrt{n}}\right ]^n. \label{eq:clup4}
\end{eqnarray}
The Lagrange dual of the above problem can be written as
\begin{eqnarray}
\min_{\x} \max_{\gamma_1} & & -\|\x\|_2 +\gamma_1\left (\mbox{max}_{\|\lambda\|_2=1}\lambda^T\left ([A \v]\begin{bmatrix}\x_{sol}-\x\\\sigma\end{bmatrix}\right )- r\right ) \nonumber \\
&& \x\in \left [-\frac{1}{\sqrt{n}},\frac{1}{\sqrt{n}}\right ]^n, \label{eq:clup5}
\end{eqnarray}
and relying on the concentration of $\gamma_1$ as
\begin{eqnarray}
\max_{\gamma_1}\min_{\x}\max_{\|\lambda\|_2=1}  & & -\|\x\|_2 +\gamma_1\lambda^T\left ([A \v]\begin{bmatrix}\x_{sol}-\x\\\sigma\end{bmatrix}\right )- \gamma_1r \nonumber \\
&& \x\in \left [-\frac{1}{\sqrt{n}},\frac{1}{\sqrt{n}}\right ]^n. \label{eq:clup5a}
\end{eqnarray}
One can then apply the RDT and proceed in a standard fashion that we outlined in \cite{StojnicCSetam09,StojnicCSetamBlock09,StojnicISIT2010binary,StojnicDiscPercp13,StojnicUpper10,StojnicGenLasso10,StojnicGenSocp10,StojnicPrDepSocp10,StojnicRegRndDlt10,Stojnicbinary16fin,Stojnicbinary16asym}. For the time being we will skip doing that and defer such a discussion for one of the later sections. Here, we will instead rely on the results that we have already created and quickly establish the solution by maximizing $c_2=\|\x\|_2^2$ so that the objective of
\begin{eqnarray}
\min_{\|\x\|_2^2=c_2}\max_{\|\lambda\|_2=1} & & \lambda^T\left ([A \v]\begin{bmatrix}\x_{sol}-\x\\\sigma\end{bmatrix} \right ) \nonumber \\
\mbox{subject to} & & \x\in \left [-\frac{1}{\sqrt{n}},\frac{1}{\sqrt{n}}\right ]^n, \label{eq:clup6}
\end{eqnarray}
remains below $r$. The main point is that the optimization in (\ref{eq:clup6}) is virtually identical to the one already considered in \cite{StojnicDiscPercp13}. For the purpose of tracking all the relevant quantities we will actually make it slightly different by adding the above mentioned $c_1=(\x_{sol})^T\x$ constraint to obtain
\begin{eqnarray}
\max_{c_2}\min_{c_1}\min_{\x}\max_{\|\lambda\|_2=1} & & \lambda^T\left ([A \v]\begin{bmatrix}\x_{sol}-\x\\\sigma\end{bmatrix} \right ) \nonumber \\
\mbox{subject to} & & \x\in \left [-\frac{1}{\sqrt{n}},\frac{1}{\sqrt{n}}\right ]^n\nonumber \\
& & (\x_{sol})^T\x=c_1\nonumber \\
& &  \|\x\|_2^2=c_2. \label{eq:clup7}
\end{eqnarray}
Before proceeding with the RDT details we will also find it convenient to define
\begin{eqnarray}
\xi_{p}(\alpha,\sigma;c_2,c_1)\triangleq \lim_{n\rightarrow\infty}\frac{1}{\sqrt{n}}\mE \min_{\x}\max_{\|\lambda\|_2=1} & & \lambda^T\left ([A \v]\begin{bmatrix}\x_{sol}-\x\\\sigma\end{bmatrix} \right ) \nonumber \\
\mbox{subject to} & & \x\in \left [-\frac{1}{\sqrt{n}},\frac{1}{\sqrt{n}}\right ]^n\nonumber \\
& & (\x_{sol})^T\x=c_1\nonumber \\
& &  \|\x\|_2^2=c_2. \label{eq:clup7a}
\end{eqnarray}

\subsubsection{Random Duality Theory -- a simple exercise}
\label{sec:rdt1}

What we will present below is basically a simple exercise within RDT and many steps can be done substantially faster. However, as it can be done through the utilization of a host of the results that we have already created we will take a moment and do it in a systematic way described in \cite{StojnicCSetam09,StojnicCSetamBlock09,StojnicISIT2010binary,StojnicDiscPercp13,StojnicUpper10,StojnicGenLasso10,StojnicGenSocp10,StojnicPrDepSocp10,StojnicRegRndDlt10,Stojnicbinary16fin,Stojnicbinary16asym}.

\vspace{.1in}
\noindent \xmyboxc{\bl{\emph{\textbf{1. First step -- \dgr{Forming the deterministic Lagrange dual} }}}}

\vspace{.1in}
We start with the first step which is just simple forming of the standard deterministic Lagrange dual of the optimization problem in (\ref{eq:clup7}) (see, e.g. (\cite{StojnicCSetam09,StojnicISIT2010binary,StojnicDiscPercp13,StojnicGenLasso10,StojnicGenSocp10,StojnicPrDepSocp10,StojnicRegRndDlt10}))
\begin{eqnarray}
\max_{c_2}\min_{c_1}\min_{\x}\max_{\|\lambda\|_2=1,\gamma,\nu} & & \lambda^T\left ([A \v]\begin{bmatrix}\x_{sol}-\x\\\sigma\end{bmatrix} \right )+\nu((\x_{sol})^T\x-c_1)+\gamma (\|\x\|_2^2-c_2) \nonumber \\
\mbox{subject to} & & \x\in \left [-\frac{1}{\sqrt{n}},\frac{1}{\sqrt{n}}\right ]^n. \label{eq:clup8}
\end{eqnarray}
As we are interested in a statistical and large dimensional scenario $\nu$ and $\gamma$ will concentrate and as scalars can be discretized and the resulting optimization over these two quantities can be taken outside
\begin{eqnarray}
\max_{c_2}\min_{c_1}\max_{\gamma,\nu}\min_{\x}\max_{\|\lambda\|_2=1} & & \lambda^T\left ([A \v]\begin{bmatrix}\x_{sol}-\x\\\sigma\end{bmatrix} \right )+\nu((\x_{sol})^T\x-c_1)+\gamma (\|\x\|_2^2-c_2) \nonumber \\
\mbox{subject to} & & \x\in \left [-\frac{1}{\sqrt{n}},\frac{1}{\sqrt{n}}\right ]^n. \label{eq:clup9}
\end{eqnarray}

\vspace{.1in}
\noindent \xmyboxc{\bl{\emph{\textbf{2. Second step -- \dgr{Forming the Random dual} }}}}

\vspace{.1in}
In the second step we introduce the auxiliary program, the so-called \bl{\textbf{random dual}} to the above primal (see, e.g. (\cite{StojnicCSetam09,StojnicISIT2010binary,StojnicDiscPercp13,StojnicGenLasso10,StojnicGenSocp10,StojnicPrDepSocp10,StojnicRegRndDlt10})). Let $\bar{{\cal X}}=\left [-\frac{1}{\sqrt{n}},\frac{1}{\sqrt{n}}\right ]^n$. Then the random dual is the following problem
\begin{equation}
\max_{c_2}\min_{c_1}\max_{\gamma,\nu}\min_{\x\in \bar{{\cal X}}}\max_{\|\lambda\|_2=1} \lambda^T\g\sqrt{\|\x_{sol}-\x\|_2^2+\sigma^2}-\|\lambda\|_2(\h^T(\x_{sol}-\x)+h_0\sigma) +\nu((\x_{sol})^T\x-c_1)+\gamma (\|\x\|_2^2-c_2),\\ \label{eq:clup10}
\end{equation}
where the components of $\g$ and $\h$ are $m$ and $n$ dimensional vectors, respectively with i.i.d. standard normal components and $h_0$ is yet another standard normal independent of all other random variables. The minus sign in front of the second term is irrelevant due to rotational symmetry of $\h$ and it is introduced to have what follows as similar as possible to some of our earlier results. Similarly to (\ref{eq:clup7a}), let $\xi_{RD}(\alpha,\sigma;c_2,c_1,\gamma,\nu)$ be the following
\begin{equation}
\lim_{n\rightarrow\infty}\frac{1}{\sqrt{n}}\mE\min_{\x\in \bar{{\cal X}}}\max_{\|\lambda\|_2=1}  \lambda^T\g\sqrt{\|\x_{sol}-\x\|_2^2+\sigma^2}-\|\lambda\|_2(\h^T(\x_{sol}-\x)+h_0\sigma) +\nu((\x_{sol})^T\x-c_1)+\gamma (\|\x\|_2^2-c_2). \label{eq:clup10a}
\end{equation}

\vspace{.1in}
\noindent \xmyboxc{\bl{\emph{\textbf{3. Third step -- \dgr{Handling the Random dual} }}}}

\vspace{.1in}
In the third step we analyze the above \bl{\textbf{random dual}}. We follow again step by step the strategy outlined in \cite{StojnicCSetam09,StojnicISIT2010binary,StojnicDiscPercp13,StojnicGenLasso10,StojnicGenSocp10,StojnicPrDepSocp10,StojnicRegRndDlt10}. It effectively boils down to the Lagrangianization and the concentration of the introduced Lagrangian slack variables. We do mention though, that in the problem at hand the first step could have been skipped as we mentioned earlier; however we have done it for the completeness as it is generally needed. Instead, one could have applied the Lagrangianization right now to arrive at (\ref{eq:clup10}). Since we have already done it we then proceed with the remaining steps. Now, we first observe that the inner optimization over $\lambda$ is trivial and one gets
\begin{eqnarray}
\max_{c_2}\min_{c_1}\max_{\gamma,\nu}\min_{\x} & & \|\g\|_2\sqrt{1-2c_1+c_2+\sigma^2}-(\h^T(\x_{sol}-\x)+h_0\sigma) +\nu((\x_{sol})^T\x-c_1)+\gamma (\|\x\|_2^2-c_2) \nonumber \\
\mbox{subject to} & & \x\in \left [-\frac{1}{\sqrt{n}},\frac{1}{\sqrt{n}}\right ]^n. \label{eq:clup11}
\end{eqnarray}
One can then follow say \cite{StojnicDiscPercp13} and define
\begin{eqnarray}
f_{box}(\h;c_2,c_1)=\max_{\gamma,\nu}\min_{\x} & & \h^T\x +\nu((\x_{sol})^T\x-c_1)+\gamma (\|\x\|_2^2-c_2) \nonumber \\
\mbox{subject to} & & \x\in \left [-\frac{1}{\sqrt{n}},\frac{1}{\sqrt{n}}\right ]^n. \label{eq:clup12}
\end{eqnarray}
Had we not introduced $c_1$ constraint with a simple shift this would be literally identical to the box constrained problem considered in \cite{StojnicDiscPercp13} and we could immediately use the solution given there. However, as mentioned earlier, here we are choosing a bit more complicated route to emphasize the structure of some of the important quantities utilized in CLuP. Still, the optimization problem in (\ref{eq:clup12}) is very similar to the one in (109) in \cite{StojnicDiscPercp13}. The solution of (\ref{eq:clup12}) is consequently very similar to (110) in \cite{StojnicDiscPercp13} with a very small change to account for $c_1$ and $\nu$. Basically, instead of (110) from \cite{StojnicDiscPercp13} one now has
\begin{eqnarray}
f_{box}(\h;c_2,c_1)  =  \max_{\gamma,\nu} & & \frac{1}{\sqrt{n}}\left (\sum_{i=1}^{n}f_{box}^{(1)}(\h_i,\gamma,\nu)\right )-\nu c_1\sqrt{n}-\gamma c_2\sqrt{n},\label{eq:clup13}
\end{eqnarray}
where
\begin{equation}
f_{box}^{(1)}(\h_i,\gamma,\nu)=\begin{cases}-|\h_i+\nu|+\gamma, & \h_i\leq -2\gamma-\nu\\
-\frac{(\h_i+\nu)^2}{4\gamma}, & -2\gamma-\nu\leq \h_i\leq 2\gamma-\nu\\
-|\h_i+\nu|+\gamma, & \h_i\geq 2\gamma-\nu,
\end{cases}\label{eq:clup14}
\end{equation}
and $\gamma$ and $\nu$ are $\sqrt{n}$ scaled versions of $\gamma$ and $\nu$ from (\ref{eq:clup12}). Moreover, the optimizing $\x_i$ is
\begin{equation}
\x_i=\frac{1}{\sqrt{n}}\min \left (\max\left (-1,-\left (\frac{\h+\nu}{2\gamma}\right )\right ),1\right ).\label{eq:clup14a}
\end{equation}
After solving the integrals one has
\begin{equation}
\mE f_{box}^{(1)}(\h_i,\gamma,\nu)=I_{22}-I_{1}+I_{21},\label{eq:clup15}
\end{equation}
where
\begin{eqnarray}
I_{22} &  = & 0.5(\nu + \gamma)\erfc((\nu + 2\gamma)/\sqrt{2}) - exp(-0.5(\nu + 2\gamma)^2)/\sqrt{2\pi} \nonumber \\
I_{1}  &  = & (\sqrt{\pi/2}(\nu^2 + 1)\erf((2\gamma - \nu)/\sqrt{2}) + \sqrt{\pi/2}(\nu^2 + 1)\erf((2\gamma + \nu)/\sqrt{2}) + exp(-0.5(\nu + 2\gamma)^2) (\nu - 2\gamma)\nonumber  \\& &- exp(-0.5(\nu-2\gamma)^2)(\nu + 2\gamma))/(4\sqrt{2\pi}\gamma)  \nonumber \\
I_{21} &  = &  -0.5(\nu - \gamma)(\erf((\nu - 2\gamma)/\sqrt{2}) + 1) - exp(-0.5(\nu - 2\gamma)^2)/\sqrt{2\pi}.
\label{eq:clup16}
\end{eqnarray}
Finally a combination of (\ref{eq:clup10})-(\ref{eq:clup16}) gives
\begin{equation}
\xi_{RD}(\alpha,\sigma;c_2,c_1,\gamma,\nu)=\sqrt{\alpha}\sqrt{1-2c_1+c_2+\sigma^2}+I_{22}-I_{1}+I_{21}-\nu c_1-\gamma c_2. \label{eq:clup17}
\end{equation}
The following theorem summarizes what we presented above.
\begin{theorem}(CLuP -- RDT estimate)
Let $\xi_{p}(\alpha,\sigma;c_2,c_1)$ and $\xi_{RD}(\alpha,\sigma;c_2,c_1,\gamma,\nu)$ be as in (\ref{eq:clup7a}) and (\ref{eq:clup17}), respectively. Then \begin{equation}
\xi_{p}(\alpha,\sigma;c_2,c_1)\geq \max_{\gamma,\nu}\xi_{RD}(\alpha,\sigma;c_2,c_1,\gamma,\nu).\label{eq:thmcluprd1}
\end{equation}
Consequently,
\begin{equation}
\min_{c_1}\xi_{p}(\alpha,\sigma;c_2,c_1)\geq \min_{c_1}\max_{\gamma,\nu}\xi_{RD}(\alpha,\sigma;c_2,c_1,\gamma,\nu).
\label{eq:thmcluprd2}
\end{equation}\label{thm:cluprd1}
\end{theorem}
\begin{proof}
Follows from the above derivation and the general RDT concepts presented in \cite{StojnicCSetam09,StojnicISIT2010binary,StojnicDiscPercp13,StojnicGenLasso10,StojnicGenSocp10,StojnicPrDepSocp10,StojnicRegRndDlt10}.
\end{proof}
Moreover, the inequalities in the above theorem are replaced with equalities when the \bl{\textbf{strong random duality}} holds. As shown in \cite{StojnicDiscPercp13,StojnicGenLasso10,StojnicGenSocp10,StojnicPrDepSocp10,StojnicRegRndDlt10} this certainly happens when the strong deterministic duality holds.

\subsubsection{CLuP's performance as a function of $r$}
\label{sec:clupfunr}

The above analysis can be utilized to do both, 1) complete the design of the CLuP and 2) characterize its performance. To complete the design of CLuP one needs to adequately choose the radius $r$. That is in general very hard task and depends on the system parameters $\alpha$ and $\sigma$ at the very least. Moreover, the dependence can be very complicated. In this introductory paper, we will try to keep things as simple and elegant as possible and will discuss only the simplest possible choices.

First, we have a firm lower bound on $r$. It is given through the following optimization
\begin{eqnarray}
r_{plt}\triangleq \lim_{n\rightarrow\infty}\frac{1}{\sqrt{n}}\mE \quad \min_{\x} & & \|\y-A\x\|_2  \nonumber \\
\mbox{subject to} & & \x\in \left [-\frac{1}{\sqrt{n}},\frac{1}{\sqrt{n}}\right ]^n. \label{eq:clup18}
\end{eqnarray}
This is of course nothing but the simple polytope relaxation of the original ML problem from (\ref{eq:ml1}). To make results easily presentable we will define
\begin{eqnarray}
r\triangleq r_{sc}r_{plt},\label{eq:clup18a}
\end{eqnarray}
where $r_{sc}$ will be the so-called scaling radius or the multiple of the minimal possible one. As mentioned earlier, the CLuP's performance can be estimated through the above mechanism relying on
\begin{eqnarray}
\max_{c_2\in[0,1]}\min_{c_1\in[0,(1+c_2)/2]}\max_{\gamma,\nu}\quad \xi_{RD}(\alpha,\sigma;c_2,c_1,\gamma,\nu)\leq r_{sc}r_{plt}. \label{eq:clup18b}
\end{eqnarray}
Moreover, when underlying functions behave nicely, one can further follow \cite{StojnicDiscPercp13,StojnicGenLasso10,StojnicGenSocp10,StojnicPrDepSocp10,StojnicRegRndDlt10} and estimate  various other performance features. For example, let $\hat{\nu}^{(CLuP)}$ be the optimal $\nu$ in (\ref{eq:clup18b}), then the probability of error $\hat{p}_{err}^{(clup)}$ is easy to obtain based on (\ref{eq:clup14a})
\begin{equation}
  \hat{p}_{err}^{(clup)}=P(\x_i\geq 0)=1-\frac{1}{2}\erfc\left ( \frac{\hat{\nu}^{(CLuP)}}{\sqrt{2}}\right ). \label{eq:clup18c}
\end{equation}
We should also add, that it is then relatively easy to see that the polytope relaxation is a trivial special case of the above formalism since for $r_{sc}=1$ one has $r=r_{plt}$ and
\begin{eqnarray}
r_{plt}= \min_{c_2\in[0,1]}\min_{c_1\in[0,(1+c_2)/2]}\xi_{p}(\alpha,\sigma;c_2,c_1). \label{eq:clup19}
\end{eqnarray}
Moreover, since (\ref{eq:clup18}) is a convex optimization problem one trivially has that the strong deterministic duality is in place which then according to \cite{StojnicDiscPercp13,StojnicGenLasso10,StojnicGenSocp10,StojnicPrDepSocp10,StojnicRegRndDlt10} implies that the strong random duality holds and consequently one has the exact equalities in (\ref{eq:thmcluprd1}) and (\ref{eq:thmcluprd2}). This then implies that
\begin{eqnarray}
r_{plt}= \min_{c_2\in[0,1]}\min_{c_1\in[0,(1+c_2)/2]}\max_{\gamma,\nu}\quad \xi_{RD}(\alpha,\sigma;c_2,c_1,\gamma,\nu). \label{eq:clup20}
\end{eqnarray}
and analogously to (\ref{eq:clup18c})
\begin{equation}
  p_{err}^{(plt)}=P(\x_i\geq 0)=1-\frac{1}{2}\erfc\left ( \frac{\hat{\nu}_{plt}}{\sqrt{2}}\right ), \label{eq:clup21}
\end{equation}
where $\hat{\nu}_{plt}$ is the optimal $\nu$ in (\ref{eq:clup20}). Of course, if one is solely interested in $r_{plt}$ and $p_{err}^{(plt)}$ they can be obtained trivially combining \cite{StojnicDiscPercp13,StojnicGenLasso10,StojnicGenSocp10,StojnicPrDepSocp10,StojnicRegRndDlt10} and in particular as an immediate consequence of the results in \cite{StojnicDiscPercp13}, most notably its equations (109) and (110).
\begin{figure}[htb]
\centering
\centerline{\epsfig{figure=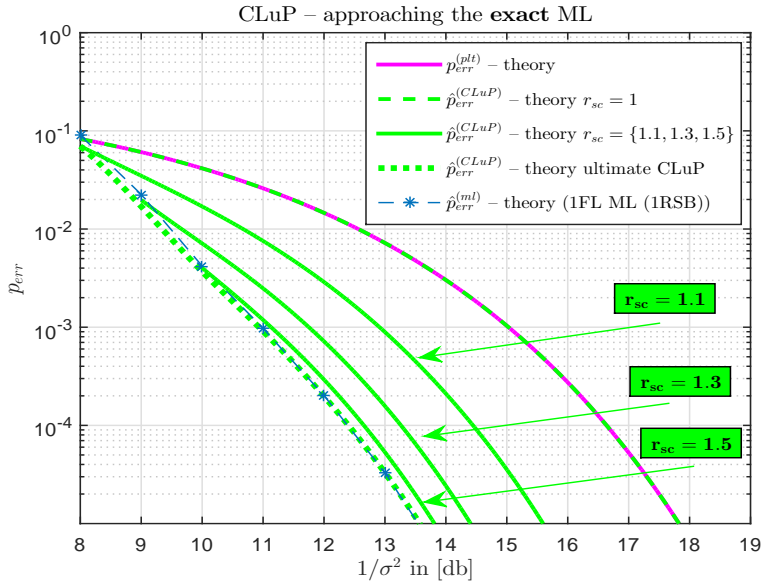,width=11.5cm,height=8cm}}
\caption{$p_{err}$ as a function of $1/\sigma^2$; $\alpha=0.8$}
\label{fig:figclup1}
\end{figure}

In Figure \ref{fig:figclup1}, we show an introductory set of results that can be obtained through the above machinery. For simplicity we focus on the probability of error $p_{err}$. We of course attack the hard regimes where traditional techniques are typically hopeless in getting anywhere close to ML. That in first place means the cases where $\alpha<1$. The plots in Figure \ref{fig:figclup1} are obtained for moderately small $\alpha=0.8$. One first observes that the curves are moving from the polytope one to the ML one as $r_{sc}$ grows. This is of course the key point. However, things are not as simple. For example, just the ML prediction itself is a notoriously hard thing to obtain. Also, at a second glance one sees that for different $r_{sc}$ the curves seem to exhibit so to say a finite domain on the left side. Moreover, the dotted green line, which will be discussed later on, appears as well and stands for the so-called ultimate level of CLuP's calculated performance. This and many other phenomena that are actually hidden behind these plots may not be easy to understand right now. In the next section we give some hints as to what is happening. However, given that this is the introductory paper on this subject we want to keep things as simple as possible and will leave more complete discussions for some of our companion papers.

\section{Discussion}
\label{sec:discussion}

We start the discussion by first noting that the upper-bound on $r$ is not as trivial as the lower bound. In fact, it seems to be strongly related to the ML curve. To fully understand this it seems that one would have to have a pretty solid understanding of the ML curve itself. This is of course one of the most challenging problems at the intersection of the signal processing and information theory. Below we start things off by first sketching what kind of estimates one obtains regarding the ML curve directly from the RTD.

\subsection{ML -- RDT estimates}
\label{sec:discussionml}

We first recall that
\begin{eqnarray}\label{eq:ml2}
\hat{\x}=\mbox{arg}\min_{\x\in{\cal X}}\|\y-A\x\|_2.
\end{eqnarray}
Merging the three RDT steps we have the following as a simple exercise.

\vspace{.1in}
\noindent \xmyboxc{\bl{\emph{\textbf{ML RDT -- three steps merged -- \dgr{Forming and handling deterministic and random duals} }}}}

\vspace{.1in}
Since now $\x\in{\cal X}$ we have $c_2=1$ and analogously to (\ref{eq:clup8}) we have as the primal version of (\ref{eq:ml2})
\begin{eqnarray}
\min_{c_1}\min_{\x}\max_{\|\lambda\|_2=1,\nu} & & \lambda^T\left ([A \v]\begin{bmatrix}\x_{sol}-\x\\\sigma\end{bmatrix} \right )+\nu((\x_{sol})^T\x-c_1) \nonumber \\
\mbox{subject to} & & \x\in {\cal X}. \label{eq:disc1}
\end{eqnarray}
Analogously to (\ref{eq:clup7a}) we will also define
\begin{eqnarray}
\xi_{p}^{(ml)}(\alpha,\sigma;c_1)\triangleq \lim_{n\rightarrow\infty}\frac{1}{\sqrt{n}}\mE \min_{\x}\max_{\|\lambda\|_2=1,\nu} & & \lambda^T\left ([A \v]\begin{bmatrix}\x_{sol}-\x\\\sigma\end{bmatrix} \right )+\nu((\x_{sol})^T\x-c_1) \nonumber \\
\mbox{subject to} & & \x\in {\cal X}. \label{eq:disc1a}
\end{eqnarray}
Following further what was done earlier we have analogously to (\ref{eq:clup10a})
\begin{equation}
\xi_{RD}^{(ml)}(\alpha,\sigma;c_1,\nu)=\lim_{n\rightarrow\infty}\frac{1}{\sqrt{n}}\mE\min_{\x\in {\cal X}}\max_{\|\lambda\|_2=1}  \lambda^T\g\sqrt{\|\x_{sol}-\x\|_2^2+\sigma^2}-\|\lambda\|_2(\h^T(\x_{sol}-\x)+h_0\sigma) +\nu((\x_{sol})^T\x-c_1). \label{eq:disc2}
\end{equation}
Optimizing over $\lambda$ and $\x$ we further have
\begin{equation}
\xi_{RD}^{(ml)}(\alpha,\sigma;c_1,\nu)=\sqrt{\alpha}\sqrt{2-2c_1+\sigma^2}-\mE|\h_i +\nu|-\nu c_1, \label{eq:disc3}
\end{equation}
where $\nu$ is $\sqrt{n}$ scaled version of $\nu$ from (\ref{eq:disc2}). Moreover, one has for the optimizing $\x_i$
\begin{equation}
\x_i=-\mbox{sign}(\h_i+\nu). \label{eq:disc3a}
\end{equation}
After solving the integral one obtains
\begin{equation}
\xi_{RD}^{(ml)}(\alpha,\sigma;c_1,\nu)=\sqrt{\alpha}\sqrt{2-2c_1+\sigma^2}-\nu c_1-(\nu\erf(\nu/\sqrt{2})-\sqrt{2/\pi}exp(-\nu^2/2)). \label{eq:disc4}
\end{equation}
Taking the derivative over $\nu$ gives
\begin{equation*}
\frac{d\xi_{RD}^{(ml)}(\alpha,\sigma;c_1,\nu)}{d\nu}= -c_1-\erf(\nu/\sqrt{2})=0, \label{eq:disc5}
\end{equation*}
and finally
\begin{equation*}
\hat{\nu}=\sqrt{2}\erfinv(-c_1). \label{eq:disc6}
\end{equation*}
Plugging this back in (\ref{eq:disc4}) we have
\begin{equation}
\xi_{RD}^{(ml)}(\alpha,\sigma;c_1)=\sqrt{\alpha}\sqrt{2-2c_1+\sigma^2}+\sqrt{2/\pi}exp(-(\sqrt{2}\erfinv(-c_1))^2/2)). \label{eq:disc7}
\end{equation}
The following theorem is in a way an ML analogue to Theorem \ref{thm:cluprd1}.
\begin{theorem}(ML -- RDT estimate)
Let $\xi_{p}^{(ml)}(\alpha,\sigma;c_1)$ and $\xi_{RD}^{(ml)}(\alpha,\sigma;c_1,\nu)$ be as in (\ref{eq:disc1a}) and (\ref{eq:disc2}) (or (\ref{eq:disc7})), respectively. Then \begin{equation}
\xi_{p}^{(ml)}(\alpha,\sigma;c_1)\geq \max_{\nu}\xi_{RD}^{(ml)}(\alpha,\sigma;c_1,\nu)=\sqrt{\alpha}\sqrt{2-2c_1+\sigma^2}+\sqrt{2/\pi}exp(-(\sqrt{2}\erfinv(-c_1))^2/2)).\label{eq:thmdiscrd1}
\end{equation}
Consequently,
\begin{equation}
\min_{c_1}\xi_{p}^{(ml)}(\alpha,\sigma;c_1)\geq \min_{c_1}\max_{\nu}\xi_{RD}^{(ml)}(\alpha,\sigma;c_1,\nu)=\min_{c_1}\sqrt{\alpha}\sqrt{2-2c_1+\sigma^2}+\sqrt{2/\pi}exp(-(\sqrt{2}\erfinv(-c_1))^2/2)).
\label{eq:thmdiscrd2}
\end{equation}\label{thm:discrd1}
\end{theorem}
\begin{proof}
Follows from the above derivation and the general RDT concepts presented in \cite{StojnicCSetam09,StojnicISIT2010binary,StojnicDiscPercp13,StojnicGenLasso10,StojnicGenSocp10,StojnicPrDepSocp10,StojnicRegRndDlt10}.
\end{proof}
One also easily has the following estimate for the probability of error
\begin{equation}
  p_{err}^{(ml)}=(1-\hat{c}_{1})/2,\label{eq:disc8}
\end{equation}
where $\hat{c}_{1}$ is the optimal $c_1$ in (\ref{eq:thmdiscrd2}). In Figure \ref{fig:figdisc1}, we show a set of results that can be obtained through the above theorem.
\begin{figure}[htb]
\centering
\centerline{\epsfig{figure=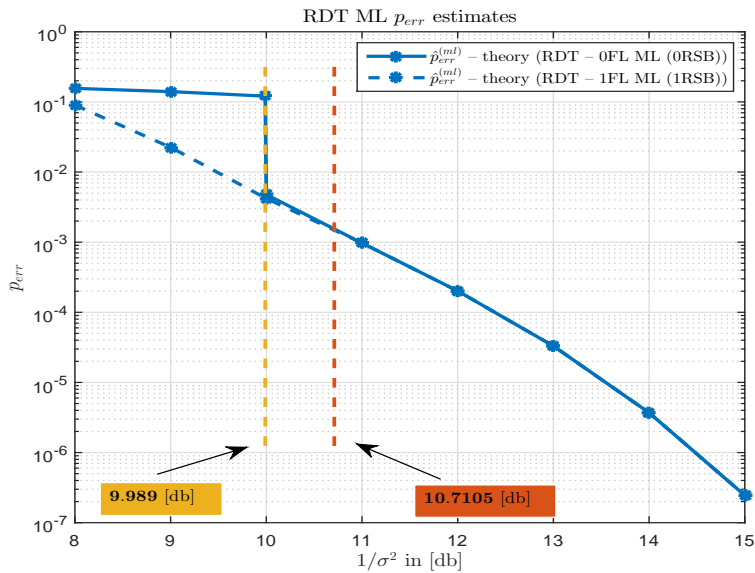,width=11.5cm,height=8cm}}
\caption{$p_{err}^{(ml)}$ as a function of $1/\sigma^2$; $\alpha=0.8$ (RDT - 0FL and 1FL (0RSB and 1RSB))}
\label{fig:figdisc1}
\end{figure}
We again focus on the probability of error $p_{err}$ and attack the same $\alpha=0.8$ regime. The full blue curve is obtained based on the above machinery. One immediately observes that the curve has a very strong and clearly visible discontinuity happening around $9.989 [db]$. This of course signals that certain corrections might be needed to the estimates that one obtains using the above theorem. We introduce these corrections through the so-called 1FL RDT (first level of full lifted random duality) and plot them as a dashed blue curve. These results are obtained through a general lifting random duality formalism that we will discuss in a separate paper. As the final results are very involved we here only draw the plot to indicate that the glitch in the original curve may indeed have to be corrected. We also mention in passing that in the companion paper we will also design a particular way of the statistical physics replica theory. It will turn out that its 1RSB version will fully match the 1FL RDT. Needless to say that the 0RSB will match the RDT prediction given above, to which we will sometimes also refer as a 0FL RDT (zeroth level of full lifted random duality, or basically just the random duality itself).

Another interesting thing is the appearance of a second vertical line around $10.71 [db]$. While it seems obvious that the corrections might be needed for $1/\sigma^2\leq 9.989[db]$ there is of course no guarantee that they may not be needed (say on a smaller scale) for the values of $1/\sigma^2$ above $9.989[db]$. The line at $10.7105 [db]$ may in fact be the critical value of $1/\sigma^2$ for which mild corrections are needed. Namely, analyzing the function $\xi_{RD}^{(ml)}(\alpha,\sigma;c_1)$ given in (\ref{eq:disc7}) one finds that it starts having multiple local minima at $1/\sigma^2=10.7105[db]$. In fact, in Figure \ref{fig:figdisc2}, we show the behavior of $\xi_{RD}^{(ml)}(\alpha,\sigma;c_1)$ at $1/\sigma^2=10.7105[db]$.
\begin{figure}[htb]
\centering
\centerline{\epsfig{figure=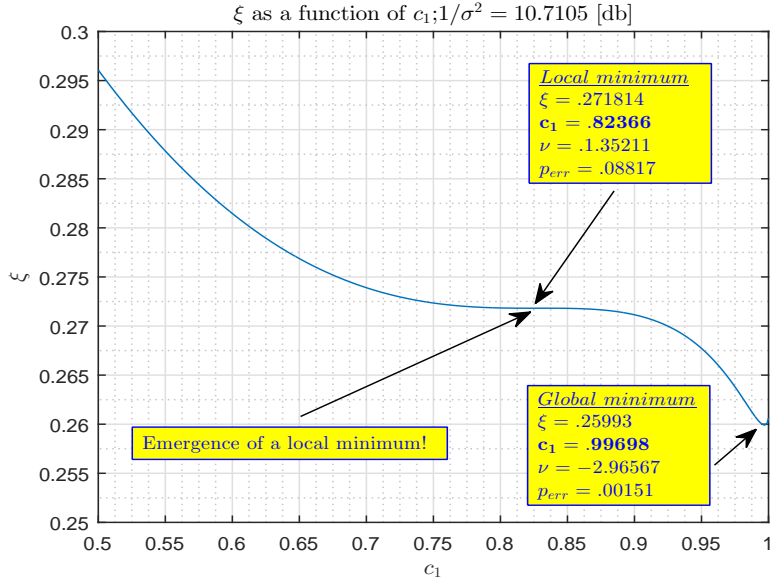,width=11.5cm,height=8cm}}
\caption{$\xi_{RD}^{(ml)}$ as a function of $c_1$; $\alpha=0.8$; $1/\sigma^2= 10.7105[db]$ (RDT - 0FL (0RSB))}
\label{fig:figdisc2}
\end{figure}
As can be seen, in addition to the global minimum at $c_1=0.99698$, one now has an emerging local minimum at $c_1=0.82366$. Moreover, the probability of error corresponding to $c_1=0.99698$ is $p_{err}^{(ml)}=0.00151$ whereas the one corresponding to $c_1=0.82366$ is $p_{err}^{(ml)}=0.08817$. This is of course a very substantial difference in performance behavior and it is directly connected to the glitch that happens at $1/\sigma^2= 9.989[db]$. Namely, as $1/\sigma^2$ moves further below $10.7105 [db]$ this local minimum becomes more and more pronounced. As Figure \ref{fig:figdisc3} indicates, it finally overtakes as the global minimum and one indeed has a very strong discontinuity.
\begin{figure}[htb]
\centering
\centerline{\epsfig{figure=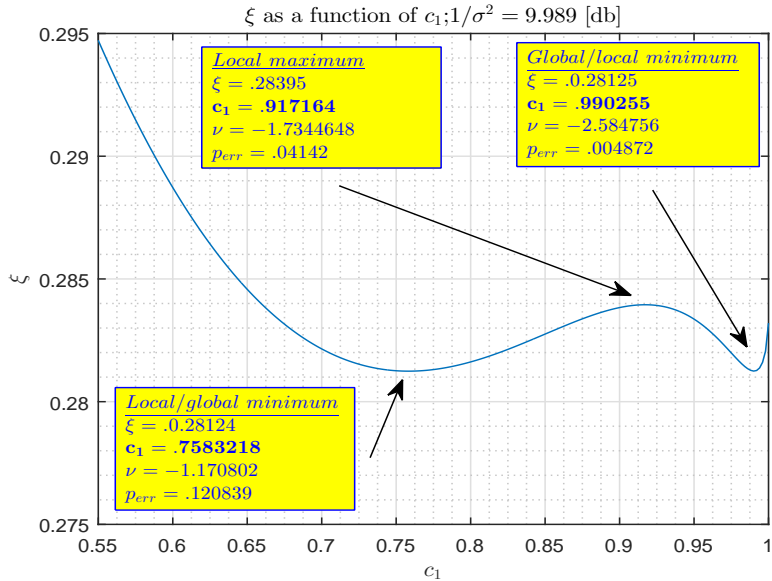,width=11.5cm,height=8cm}}
\caption{$\xi_{RD}^{(ml)}$ as a function of $c_1$; $\alpha=0.8$; $1/\sigma^2= 9.989[db]$ (RDT - 0FL (0RSB))}
\label{fig:figdisc3}
\end{figure}
While we leave the details of the 1FL RDT for the companion papers we do mention here that it substantially smoothens the glitch. Still, we do believe that higher levels of lifting actually achieve the exact performance (in fact, the second level is probably already getting close enough that visually distinguishing further improvements would be virtually impossible). However, from the practical viewpoint (and as we will see later on when we discuss the numerical results) the corrections at 1FL RDT are already very close to the simulated values. In fact, for $1/\sigma^2=10[db]$ the correction does exist but it is fairly small (virtually invisible in Figure \ref{fig:figdisc1}). On the other hand already for $1/\sigma^2=11[db]$ we were not able to find any noticeable corrections. This may indicate that the line of mild or no corrections might indeed be somewhere between $10-11[db]$ (as mentioned above, quite possibly maybe not even far away from $10.7105[db]$). From this small discussion one can already see that the whole story is way more complicated compared to how it may initially seem from the nice plots. This type of discussion is basically provided just as a hint as to what kind of miracles might be happening and how they may be related to CLuP which is the main interest of this paper. We of course leave more thorough discussions regarding the ML performance for one of our companion papers.

\subsection{CLuP -- how it relates to ML}
\label{sec:discussionclup}

Now that we did get a bit of a feeling as to what happens with ML performance we will get back to the CLuP itself. We recall that the plots in Figure \ref{fig:figclup1} seem to exhibit a finite domain on the left side, meaning that below certain values of SNR $1/\sigma^2$, ceratin scalings of $r_{plt}$ might not be possible. We also recall the existence of a dashed green curve in Figure \ref{fig:figclup1}. These things are to a large degree connected to the ML performance and we will discuss them in a bit more detail  below. However, before doing so, we also observe several properties of CLuP $\xi_{RD}$ function that in a way may also be connected to the above discussed ML performance.

\subsubsection{CLuP -- $\xi_{RD}$ local optima}
\label{sec:discussioncluplocal}

We will focus on the SNR regime where the above discussion indicates that the ML corrections might be needed. So, we first start with $1/\sigma^2=11[db]$ (this is actually slightly above the above discussed $10.7105[db]$ line but it is a good starting point). We select a particular value $c_2=0.9979$ (this choice will become clear later on) and show in Figure \ref{fig:figdisc4} how the $\xi_{RD}$ changes as a function of $c_1$.
\begin{figure}[htb]
\centering
\centerline{\epsfig{figure=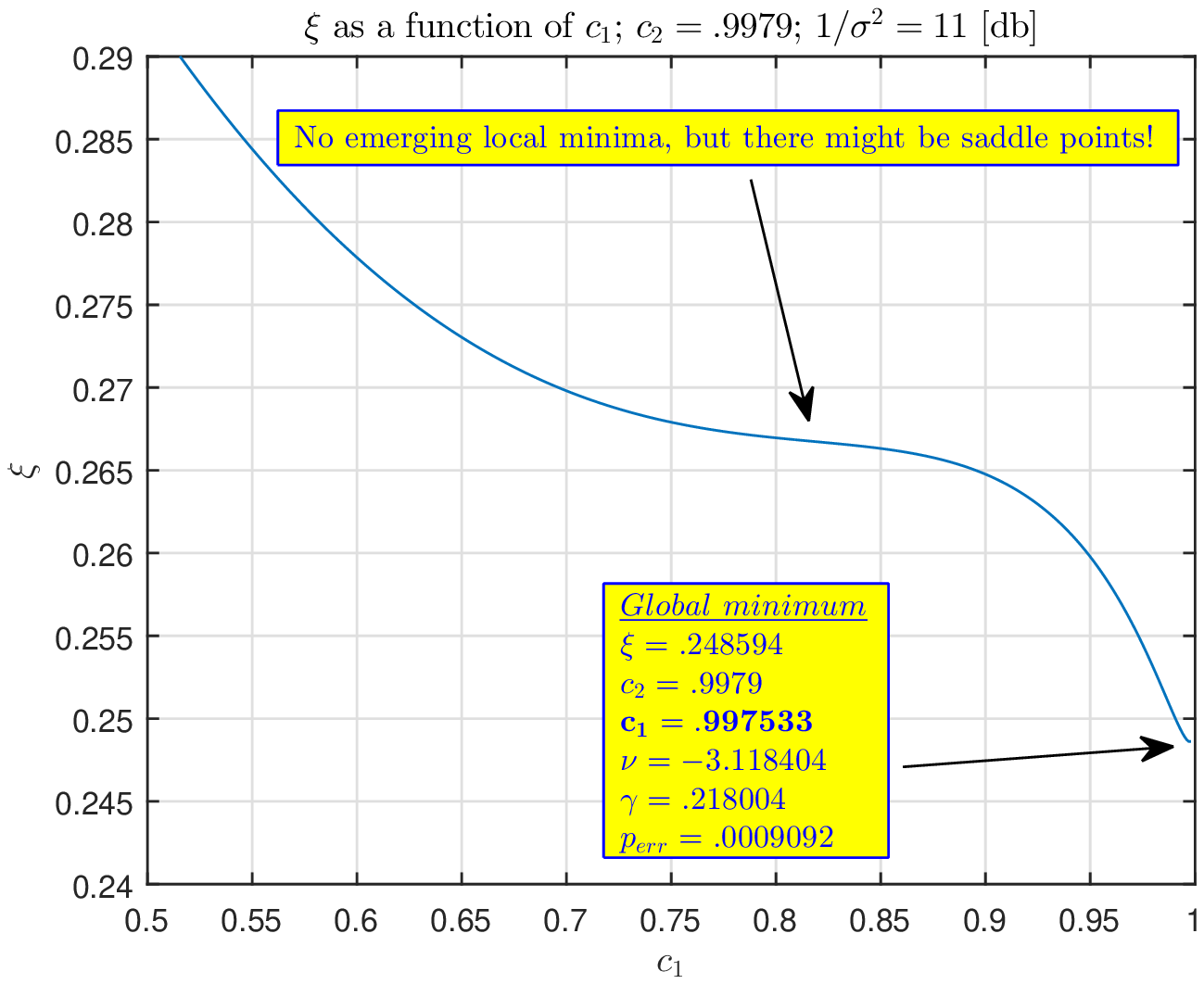,width=11.5cm,height=8cm}}
\caption{$\xi_{RD}$ as a function of $c_1$; $\alpha=0.8$; $1/\sigma^2=11[db]$; $c_2=0.9979$ (RDT - 0FL (0RSB))}
\label{fig:figdisc4}
\end{figure}
As it turns out there is no an emerging local optimum (based on the shape of the curve, one might hypothetically assume that there might be some saddle points; however, given how complicated the underlining functions are this may seem rather unlikely). This is of course only a particular choice of $c_2$ which will correspond to a particular choice of $r_{sc}$ and consequently $r$. However, we found no $c_2$ where a local optimum over $c_1$ emerges. One can now note that this is in a nice agreement with the above ML discussion.

On the other hand, things are a little different as one moves the SNR down to $10[db]$.
\begin{figure}[htb]
\centering
\centerline{\epsfig{figure=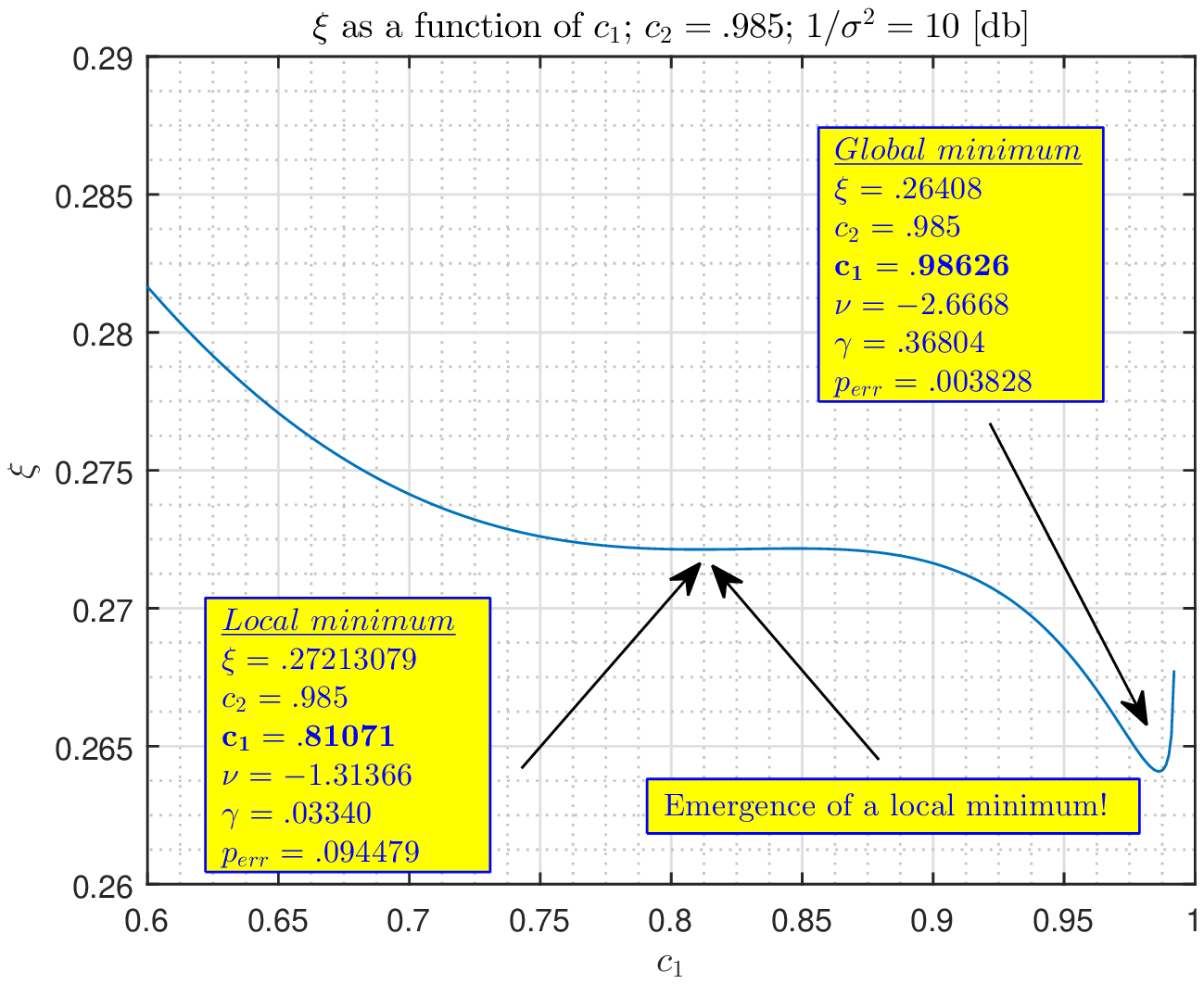,width=11.5cm,height=8cm}}
\caption{$\xi_{RD}$ as a function of $c_1$; $\alpha=0.8$; $1/\sigma^2=10[db]$; $c_2=0.985$ (RDT - 0FL (0RSB))}
\label{fig:figdisc5}
\end{figure}
In Figure \ref{fig:figdisc5} we show how $\xi_{RD}$ changes as a function of $c_1$ for $c_2=0.985$ and observe the emergence of a local minimum. Based on the optimizing values one again notes a very sharp difference in the estimated probabilities of error. However, we found no values for $c_2$ where the emerging local optimum overtakes and becomes the global minimum. This might indicate that since it is an iterative algorithm, CLuP may have problems getting to the global optimum but with a careful strategy might be able to avoid local traps as well.

As one moves the SNR further down to $9[db]$ things become even more different.
\begin{figure}[htb]
\centering
\centerline{\epsfig{figure=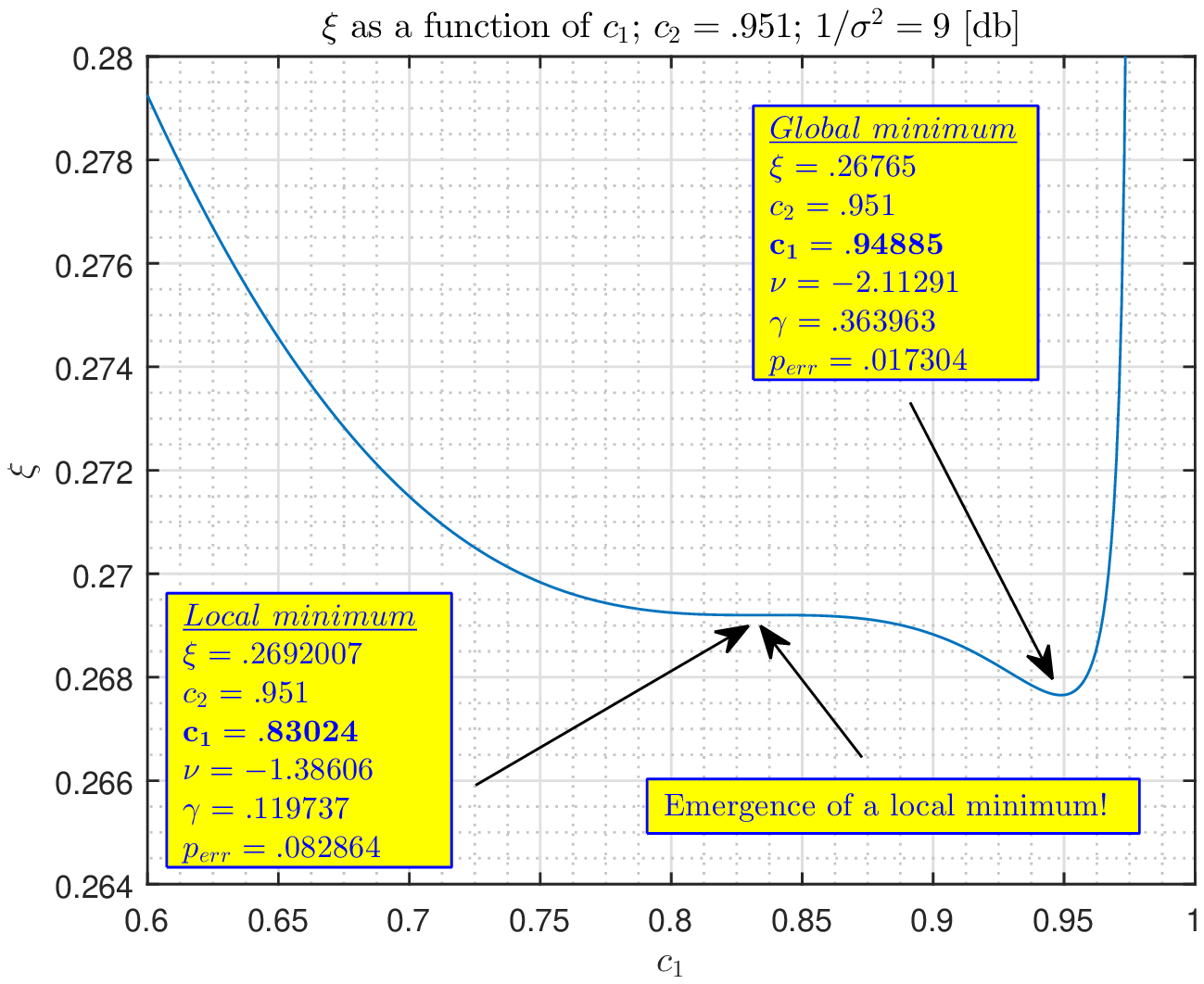,width=11.5cm,height=8cm}}
\caption{$\xi_{RD}$ as a function of $c_1$; $\alpha=0.8$; $1/\sigma^2=9[db]$; $c_2=0.951$ (RDT - 0FL (0RSB))}
\label{fig:figdisc6}
\end{figure}
First, in Figure \ref{fig:figdisc6} we show the behavior for $c_2=0.951$ and observe the emergence of a local minimum. Then, in Figure \ref{fig:figdisc7} we show the behavior for $c_2=0.96$ and observe that the local minimum overtakes as the global. This actually might pose a serious problem for success of CLuP.
\begin{figure}[htb]
\centering
\centerline{\epsfig{figure=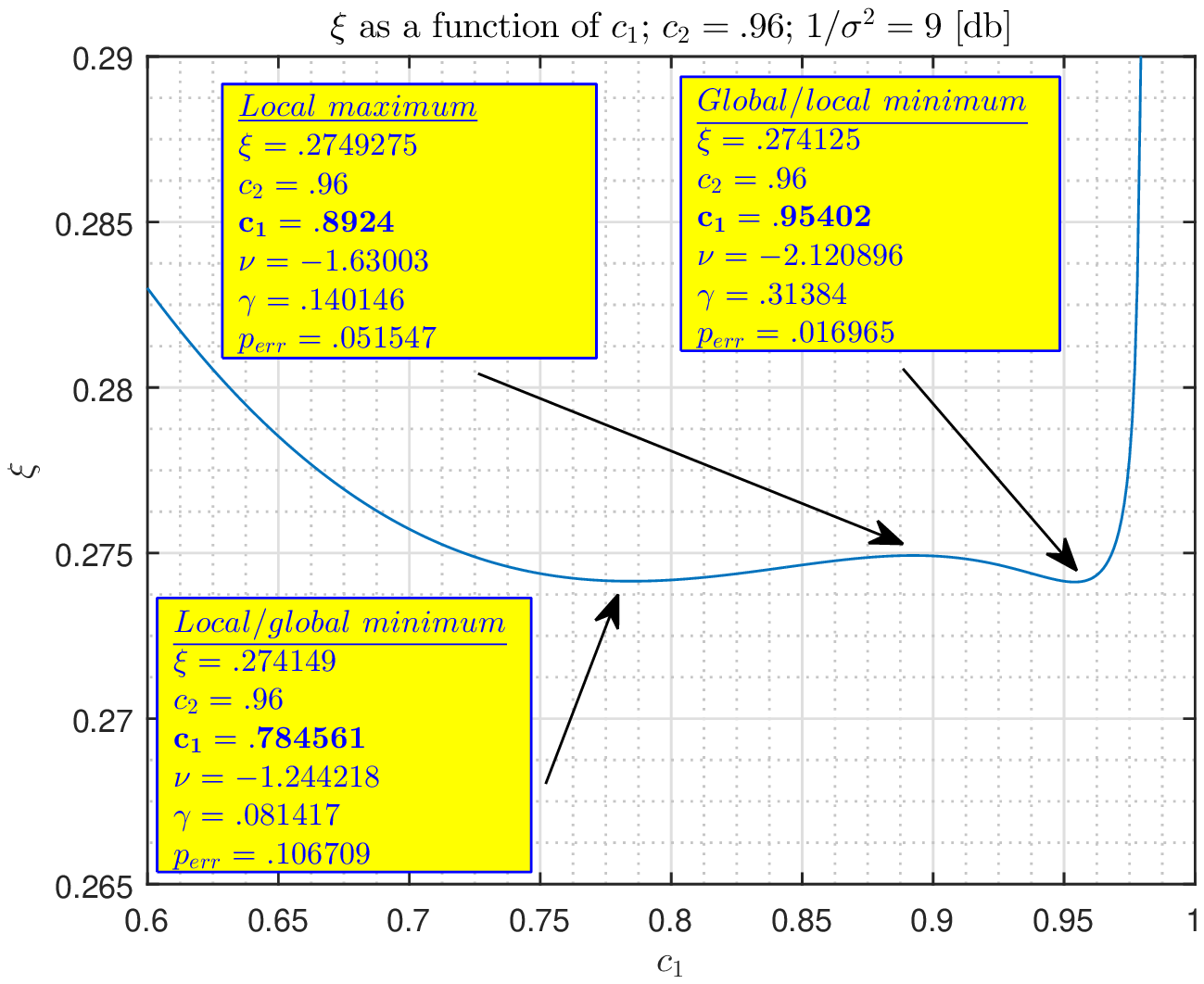,width=11.5cm,height=8cm}}
\caption{$\xi_{RD}$ as a function of $c_1$; $\alpha=0.8$; $1/\sigma^2=9[db]$; $c_2=0.96$ (RDT - 0FL (0RSB))}
\label{fig:figdisc7}
\end{figure}
Finally, in Figure \ref{fig:figdisc8} we show the behavior for $c_2=0.975$ and observe the disappearance of a desired local minimum which might put CLuP in a position of no success.
\begin{figure}[htb]
\centering
\centerline{\epsfig{figure=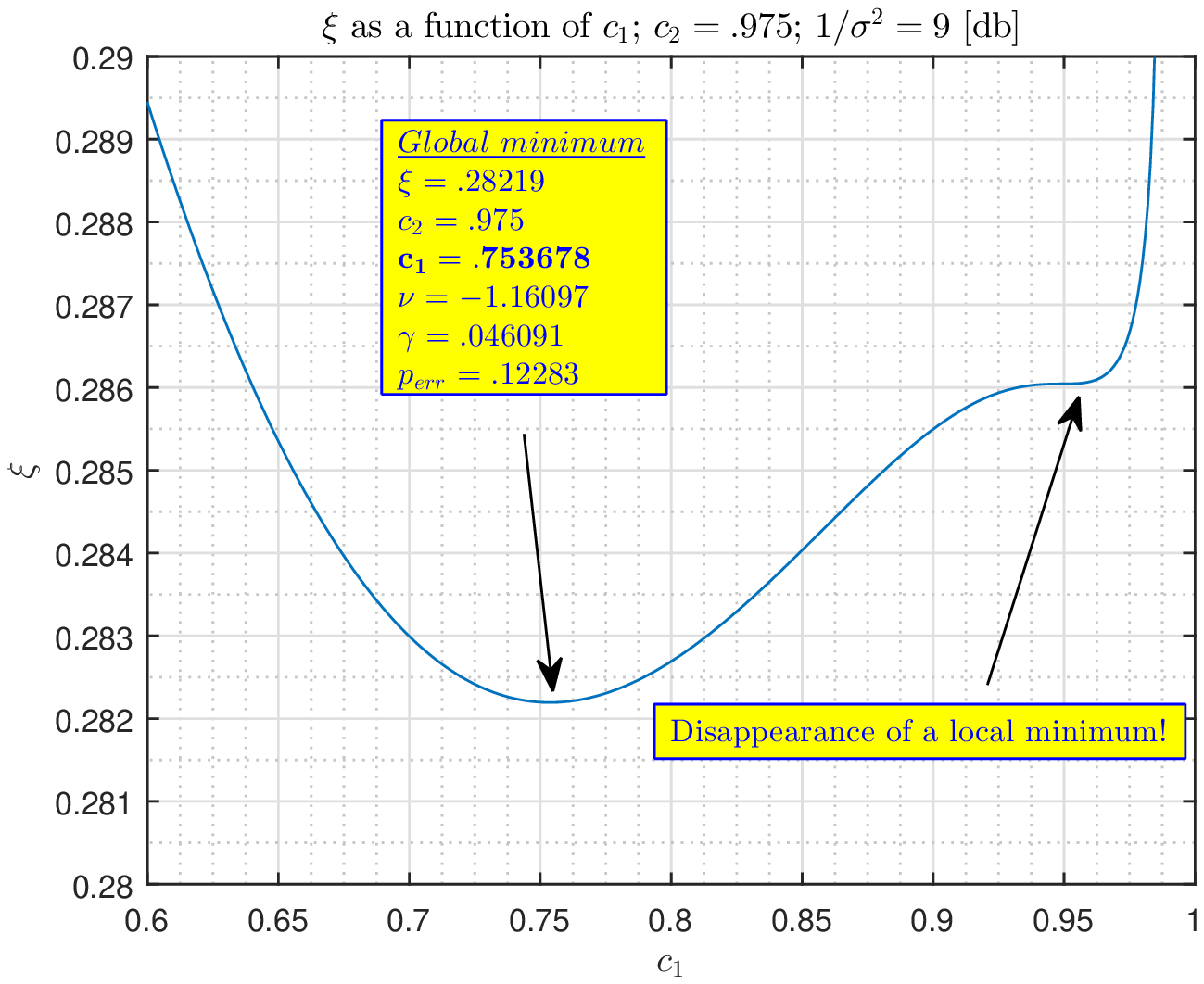,width=11.5cm,height=8cm}}
\caption{$\xi_{RD}$ as a function of $c_1$; $\alpha=0.8$; $1/\sigma^2=9[db]$; $c_2=0.975$ (RDT - 0FL (0RSB))}
\label{fig:figdisc8}
\end{figure}
These are some interesting properties of $\xi_{RD}$. One should keep in mind though that the choice of $r$ might be such that the optimal $c_2$ is not into the range where the above discussed properties of $\xi_{RD}$ happen. Plus, one should of course always keep in mind that this is in the regime below the above discussed line of corrections where various miracles are possible which can cause the properties of $\xi_{RD}$ to change.

\subsubsection{CLuP -- $\xi_{RD}$ stationary points}
\label{sec:discussionclupstatpoints}

While the above discussion goes into tiny details to understand particular role of all key parameters, here we would like to emphasize that for the completeness we have also proceeded in the standard RDT fashion mentioned right after (\ref{eq:clup5a}). As one recalls, in (\ref{eq:clup5a}) we had
\begin{eqnarray}
\max_{\gamma_1}\min_{\x}\max_{\|\lambda\|_2=1}  & & -\|\x\|_2 +\gamma_1\lambda^T\left ([A \v]\begin{bmatrix}\x_{sol}-\x\\\sigma\end{bmatrix}\right )- \gamma_1r \nonumber \\
&& \x\in \left [-\frac{1}{\sqrt{n}},\frac{1}{\sqrt{n}}\right ]^n. \label{eq:discclup5a}
\end{eqnarray}
Combining this with Theorem \ref{thm:cluprd1} (and the analysis that preceded Theorem \ref{thm:cluprd1}) and in particular with (\ref{eq:clup17}) one has
\begin{equation}
\xi_{RD,\gamma_1}(\alpha,\sigma;c_2,c_1,\gamma,\nu)=-\sqrt{c_2}+\gamma_1(\sqrt{\alpha}\sqrt{1-2c_1+c_2+\sigma^2}+I_{22}-I_{1}+I_{21}-\nu c_1-\gamma c_2)-\gamma_1 r, \label{eq:discclup17}
\end{equation}
where
$I_{22}$, $I_{1}$, and $I_{21}$ are as given in (\ref{eq:clup16}). One can then utilize the following set of equations
\begin{eqnarray}\label{eq:discclup17a}
  \frac{d\xi_{RD,\gamma_1}(\alpha,\sigma;c_2,c_1,\gamma,\nu)}{d c_2} & = &  0\nonumber \\
  \frac{d\xi_{RD,\gamma_1}(\alpha,\sigma;c_2,c_1,\gamma,\nu)}{d c_1} & = &  0\nonumber \\
  \frac{d\xi_{RD,\gamma_1}(\alpha,\sigma;c_2,c_1,\gamma,\nu)}{d\nu} & = &  0\nonumber \\
  \frac{d\xi_{RD,\gamma_1}(\alpha,\sigma;c_2,c_1,\gamma,\nu)}{d\gamma} & = &  0\nonumber \\
  \frac{d\xi_{RD,\gamma_1}(\alpha,\sigma;c_2,c_1,\gamma,\nu)}{d\gamma_1} & = &  0.
\end{eqnarray}
After solving over $\nu$ and $\gamma_1$ things can be a bit simplified since
\begin{eqnarray}
  \nu &=& -2\sqrt{\alpha}/2/\sqrt{1-2c_1+c_2+\sigma^2}\nonumber \\
  \gamma_1 &=& 1/2/\sqrt{c_2}/(-\nu/2-\gamma).\label{eq:discclup17a1}
\end{eqnarray}
Finally, after solving over $c_2$, $c_1$, and $\gamma$ we find the following two solutions for $1/\sigma=10$[db]
\begin{eqnarray}
\xi_{RD}& = & 0.225173, c_2=0.46075, c_1=0.56459, \nu=-1.361508, \gamma=1.10981, \gamma_1=-1.716832\nonumber \\
\xi_{RD}& = & 0.225173, c_2=0.93035, c_1=0.94857, \nu=-2.450658, \gamma=0.68036, \gamma_1=0.9511982, \label{eq:discclup17b}
\end{eqnarray}
and the following two for $1/\sigma=9$[db]
\begin{eqnarray}
\xi_{RD}& = & 0.252694, c_2=0.43726, c_1=0.53669, \nu=-1.278041, \gamma=1.10130, \gamma_1=-1.635647\nonumber \\
\xi_{RD}& = & 0.252694, c_2=0.92731, c_1=0.93236, \nu=-2.060218, \gamma=0.45413, \gamma_1=0.901472. \label{eq:discclup17c}
\end{eqnarray}
We have found no other stationary points and the above two actually exactly correspond to the two shown later on in Figures \ref{fig:figdisc12} and \ref{fig:figdisc12a}.

\subsubsection{CLuP -- limiting $r$ through objective values}
\label{sec:discussioncluplimrobj}

Now we finally get to address the existence of a finite domain barrier on the left side of different $r_{sc}$ plots in Figure \ref{fig:figclup1}. Basically as plots indicate, below certain values of SNR $1/\sigma^2$, ceratin scalings of $r_{plt}$ might not be possible. This is of course directly related to the above discussion about the ML performance. Namely, as $r_{sc}$ (and consequently $r$) grows, the CLuP optimal $c_2$ grows as well. Due to the CLuP's structure $c_2$ can not grow above $1$. This in turn effectively imposes the limit on $r$ and $r_{sc}$ (of course, both, $r$ and $r_{sc}$ are basically without upper limits; however, raising them above ceratin values may be useless for the whole CLuP concept). What one might expect is that when optimal $c_2=1$ is such that the achieving $r$ is matching the optimal $\xi_{p}^{(ml)}(\alpha,\sigma)$ (obtained after the optimization over $c_1$) then CLuP's performance matches ML. Or alternatively, when one switches to the RDT terrain, one might expect that when optimal $c_2=1$ is such that the achieving $r$ is matching the optimal $\xi_{RD}^{(ml)}(\alpha,\sigma)$ (obtained again after the optimization over $c_1$) then CLuP's performance matches ML. This though may not even be the best one can do. However, before getting to this we first in Figure \ref{fig:figdisc10} show the limiting upper values that $r_{sc}$ can take based on the above reasoning. Namely, as the above suggests one can have as the upper limit $r_{sc}=\xi_{p}^{(ml)}/r_{plt}$. On the other hand, as we have mentioned when discussing the ML performance, in certain range of SNR one might need to correct the values for $\xi_{p}^{(ml)}$. For such a correction we utilize the values obtained through the 1FL RDT and refer to them as $\xi_{p}^{(1FL ML)}$.
\begin{figure}[htb]
\centering
\centerline{\epsfig{figure=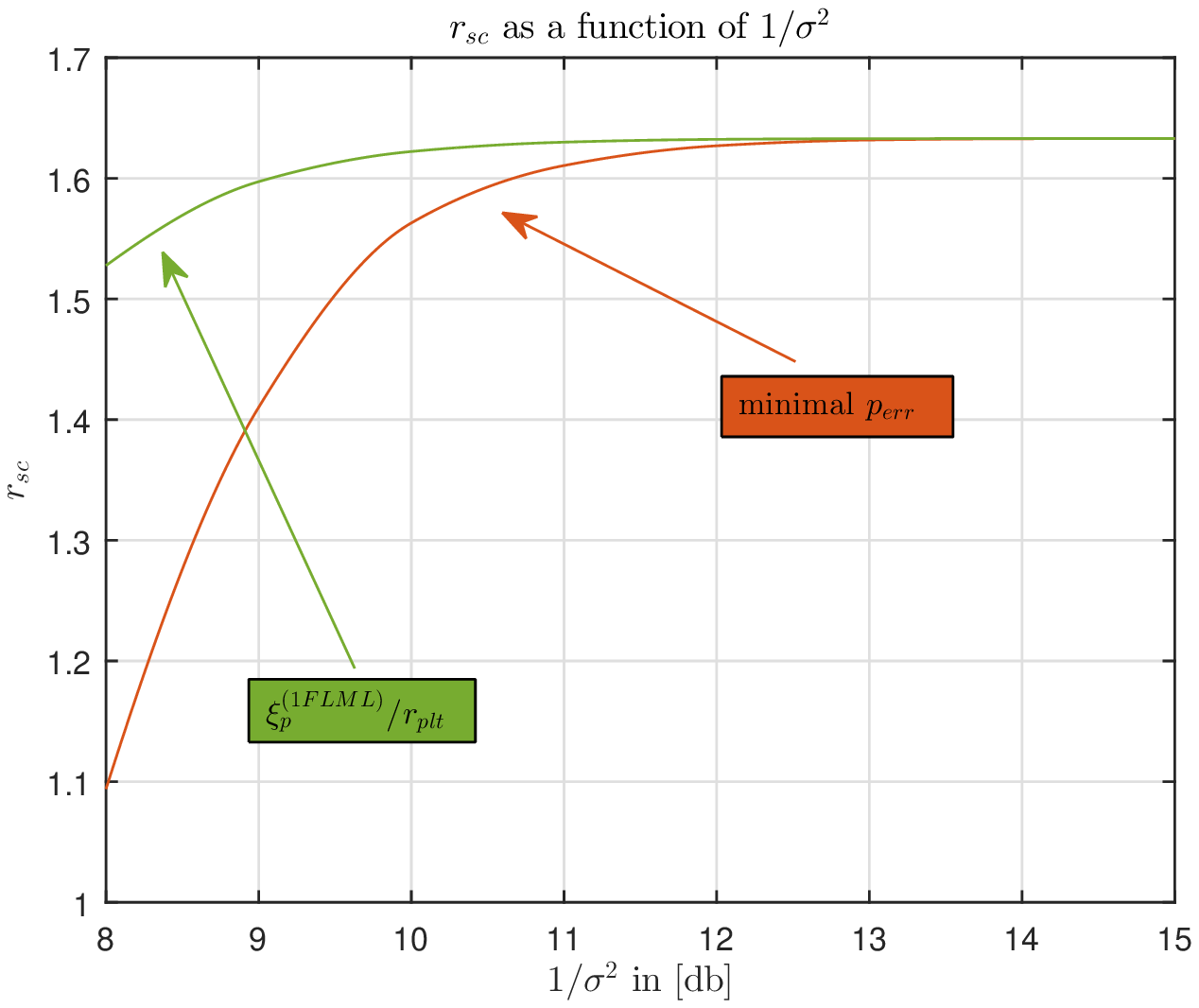,width=11.5cm,height=8cm}}
\caption{$r_{sc}=\xi_{p}^{(1FL ML)}/r_{plt}$ or is chosen so that $p_{err}^{(CLuP)}$ is minimal and given as a function of $1/\sigma^2$; $\alpha=0.8$}
\label{fig:figdisc10}
\end{figure}
It is interesting to note that based on Figure \ref{fig:figdisc10} some values of $r_{sc}$ might be restricted, but all the three values discussed earlier in Figure \ref{fig:figclup1} remain permissible. In the following subsection we discuss a different limiting strategy.

\subsubsection{CLuP -- limiting $r$ through minimal $p_{err}$}
\label{sec:discussioncluplimrperr}

The above choice of limiting $r$ seems reasonable (in fact when it comes to achieving ML may be the most reasonable). However, one can ignore ML for a moment and wonder what would be the best way to design CLuP so that it achieves the best possible performance. The immediate question would be what would be the criteria to determine what the best possible performance is. There are of course many criteria that one can consider but if we just stick with the probability of error $p_{err}^{(CLuP)}$ then seemingly the most natural way would be to choose $r_{sc}$ as to minimize $p_{err}^{(CLuP)}$. Recalling on (\ref{eq:clup18c}), this essentially means that one should choose $r_{sc}$ so that $\hat{\nu}^{(CLuP)}$ is minimized. The results that we obtained following this strategy are shown in Figures \ref{fig:figdisc10} and \ref{fig:figdisc11a} (the choice $c_2=0.9979$ mentioned earlier is now clear from Figure \ref{fig:figdisc11a}). Moreover, the resulting $p_{err}^{(CLuP)}$ is exactly the dashed curve in Figure \ref{fig:figclup1} to which we refer as the ultimate CLuP calculated performance. As Figure \ref{fig:figclup1} is mainly concerned with the effect of changing $r_{sc}$ rather than with this type of subtlety, we below in
Figure \ref{fig:figdisc11} show once again this curve together with the ML one obtained through 1FL RDT.
\begin{figure}[htb]
\centering
\centerline{\epsfig{figure=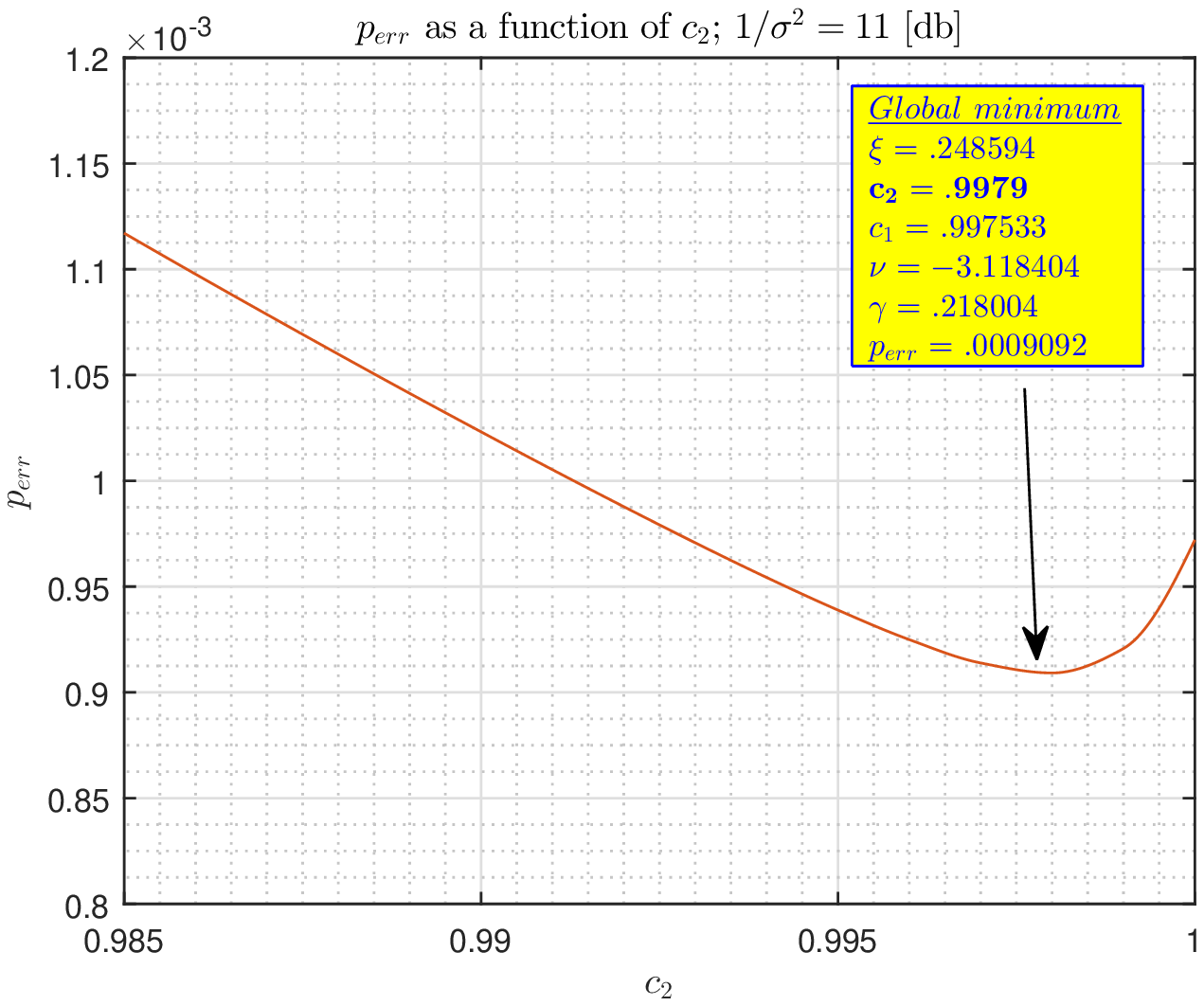,width=11.5cm,height=8cm}}
\caption{$p_{err}$ as a function of $c_2$; $1/\sigma^2=11$[db]; $\alpha=0.8$}
\label{fig:figdisc11a}
\end{figure}
\begin{figure}[htb]
\centering
\centerline{\epsfig{figure=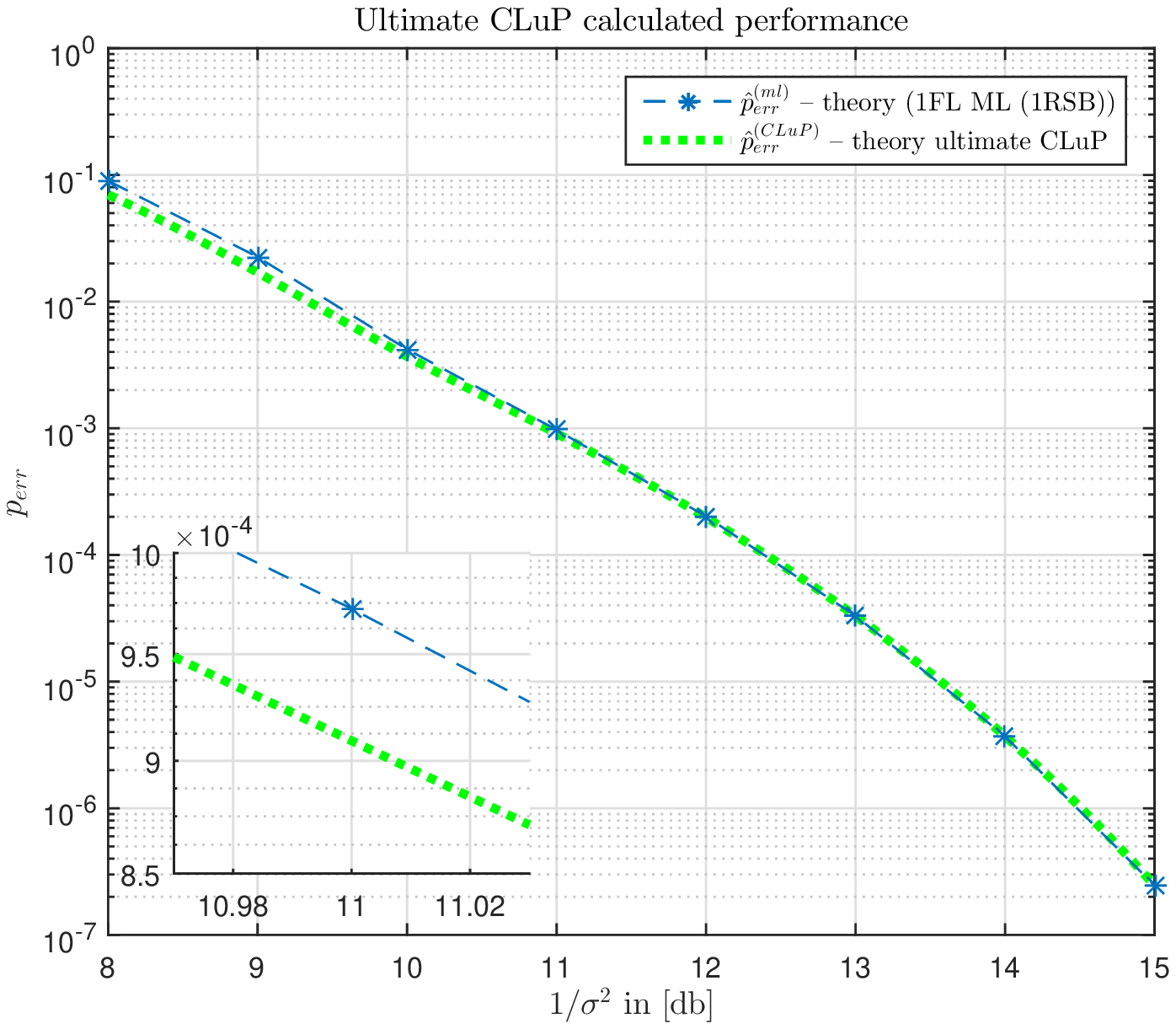,width=11.5cm,height=8cm}}
\caption{$r_{sc}$ chosen such that $p_{err}^{(CLuP)}$ is minimal and given as a function of $1/\sigma^2$; $\alpha=0.8$}
\label{fig:figdisc11}
\end{figure}
Since the curves are close to each other we also provide some of the numerical values in Table \ref{tab:tabclup1}. The comparison of the values is not so interesting in the regime where $ 1/\sigma^2 \leq 10.7105$[db] as one may expect further corrections to $\hat{p}_{err}^{(ml)}$ (we do not believe that they are significant but, as we will see later on when discussing results obtained from numerical experiments, they are likely to push $\hat{p}_{err}^{(ml)}$ a bit below the values given for $\hat{p}_{err}^{(CLuP)}$) and $\hat{p}_{err}^{(CLuP)}$ may need to be readjusted as discussed below depending on the way how the appearance of local optima is handled.  It is interesting though that for $ 1/\sigma^2 \geq 10.7105$[db] (where one expects no or very mild further corrections, quite possibly in some regimes both visually and computationally not detectable) $\hat{p}_{err}^{(CLuP)}$ remains below $\hat{p}_{err}^{(ml)}$.
\begin{table}[h]
\caption{Numerical values for $\hat{p}_{err}^{(ml)}$ and $\hat{p}_{err}^{(CLuP)}$ that correspond to the data in Figure \ref{fig:figdisc11}}\vspace{.1in}
\hspace{-0in}\centering
\footnotesize{
\begin{tabular}{||c||c|c|c|c|c|c|c|c||}\hline\hline
$ 1/\sigma^2 $[db] & $8  $ & $9  $ & $10  $ & $11  $ & $12  $ & $13  $ & $14  $ & $15  $ \\ \hline\hline
$\hat{p}_{err}^{(ml)}$ & $9.00e-02  $ & $2.25e-02  $ & $4.20e-03  $ & $9.72e-04  $ & $2.01e-04  $ & $3.30e-05  $ & $3.70e-06  $ & $2.46e-07  $ \\ \hline
$\hat{p}_{err}^{(CLuP)}$ & $6.98e-02  $ & $1.70e-02  $ & $3.69e-03  $ & $9.09e-04  $ & $1.97e-04  $ & $3.29e-05  $ & $3.70e-06  $ & $2.46e-07  $ \\ \hline\hline
\end{tabular}}
\label{tab:tabclup1}
\end{table}

\subsubsection{CLuP -- limiting $r$ through appearance of local/global optima}
\label{sec:discussioncluplimrlocglobopt}

In Figures \ref{fig:figdisc12} and \ref{fig:figdisc12a} we present $\xi$ as a function of $c_2$. One can now clearly see the appearance of the local optima which for $1/\sigma^2=9$ even overtake as global optima. Moreover, one can restrict $r_{sc}$ so that these regimes are not reached.
\begin{figure}[htb]
\begin{minipage}[b]{.5\linewidth}
\centering
\centerline{\epsfig{figure=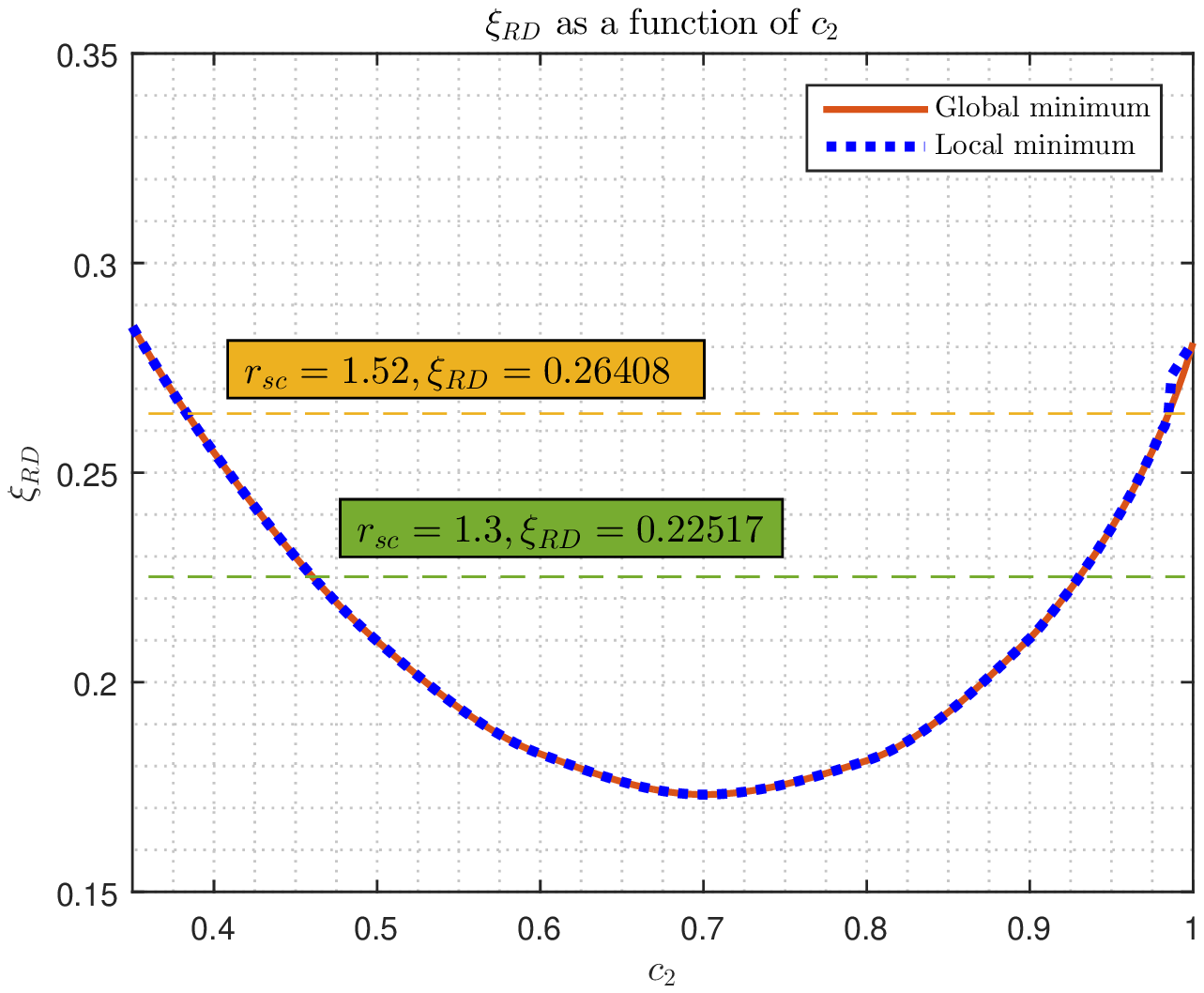,width=8cm,height=6cm}}
\end{minipage}
\begin{minipage}[b]{.5\linewidth}
\centering
\centerline{\epsfig{figure=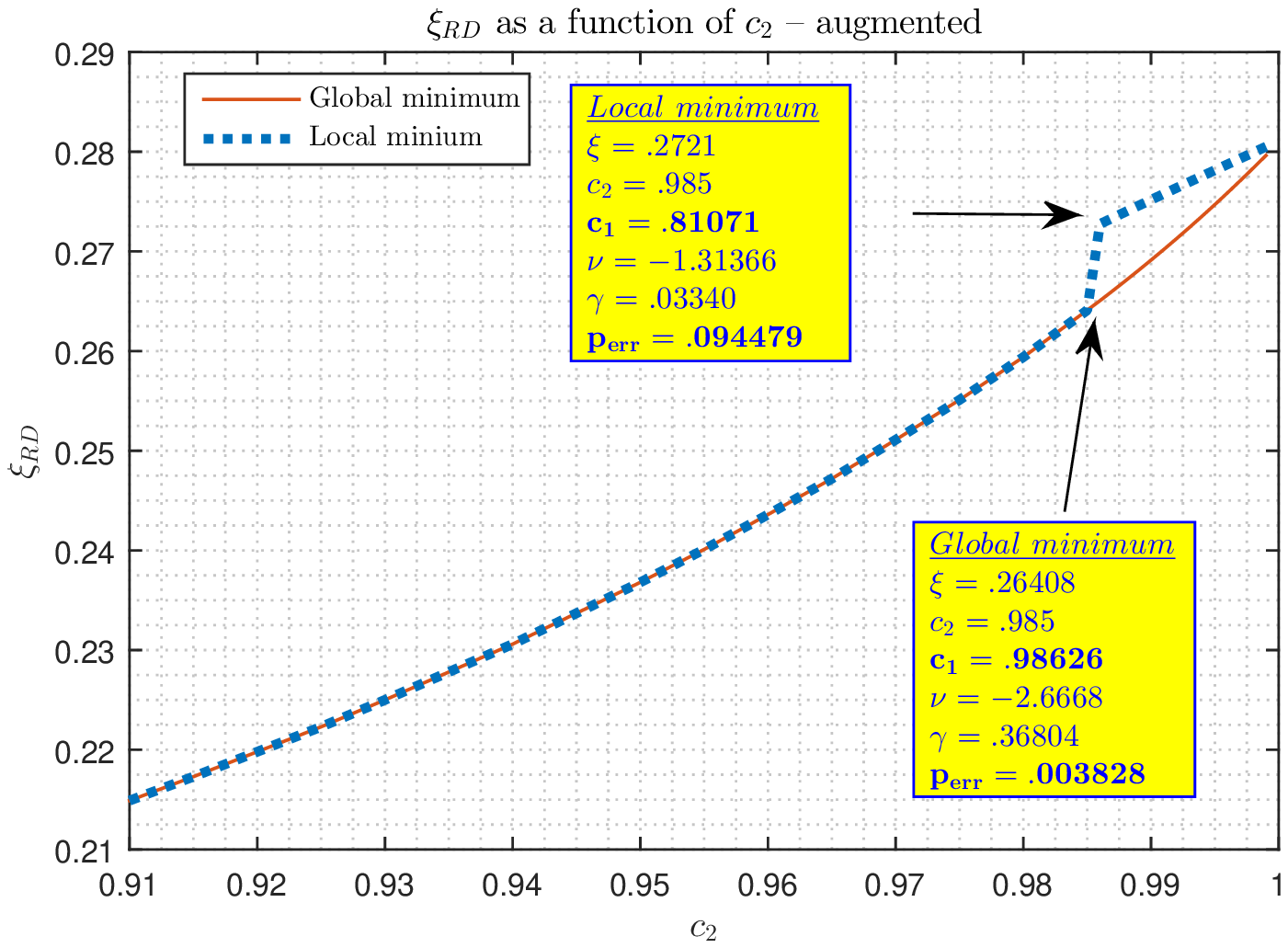,width=8cm,height=6cm}}
\end{minipage}
\vspace{-0.3in}
\caption{$\xi_{RD}$ as a function of $c_2$; $1/\sigma^2=10$ db a) left - full range of $c_2$, b) right - appearance of local optima}
\label{fig:figdisc12}
\end{figure}

\begin{figure}[htb]
\begin{minipage}[b]{.5\linewidth}
\centering
\centerline{\epsfig{figure=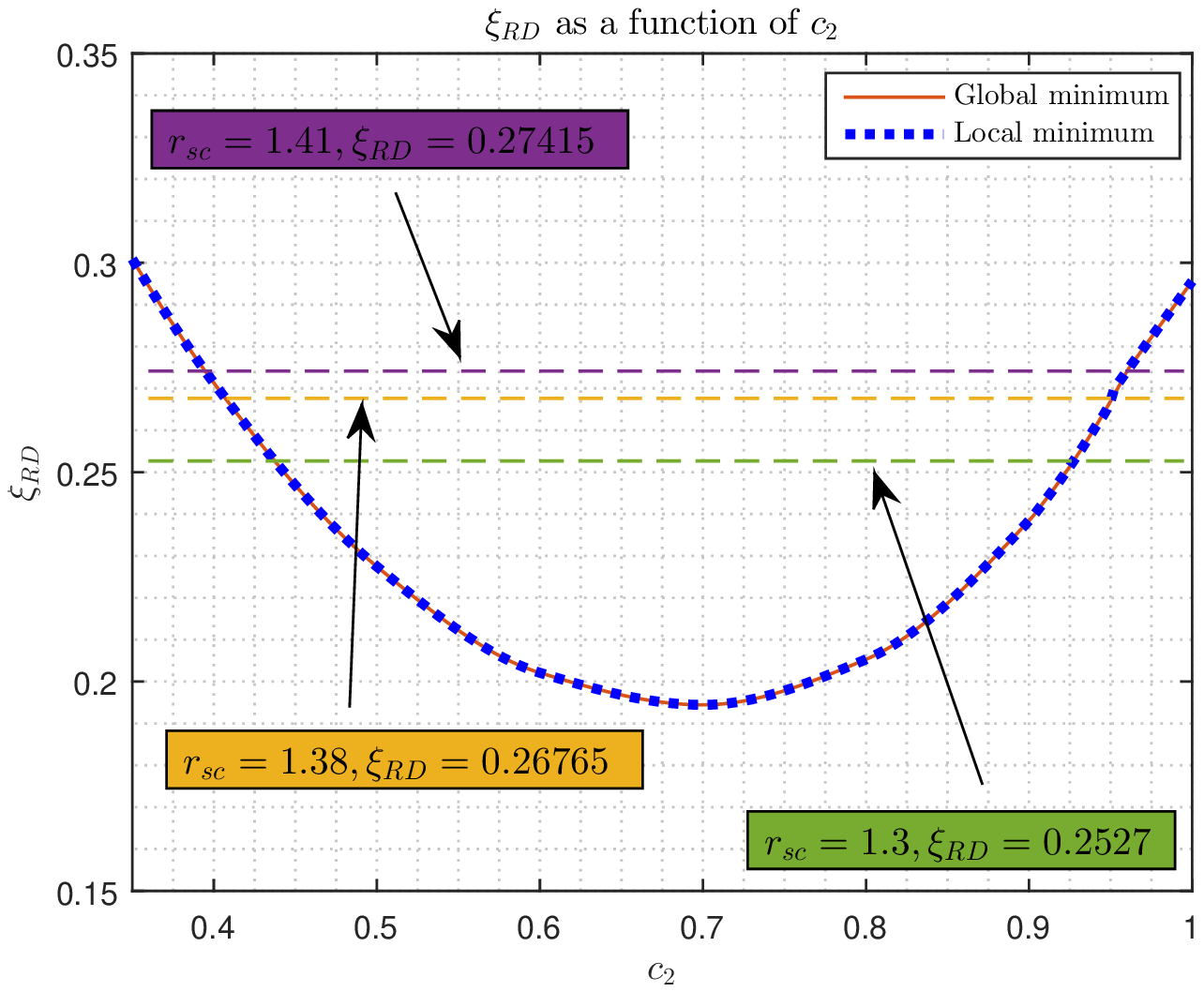,width=8cm,height=6cm}}
\end{minipage}
\begin{minipage}[b]{.5\linewidth}
\centering
\centerline{\epsfig{figure=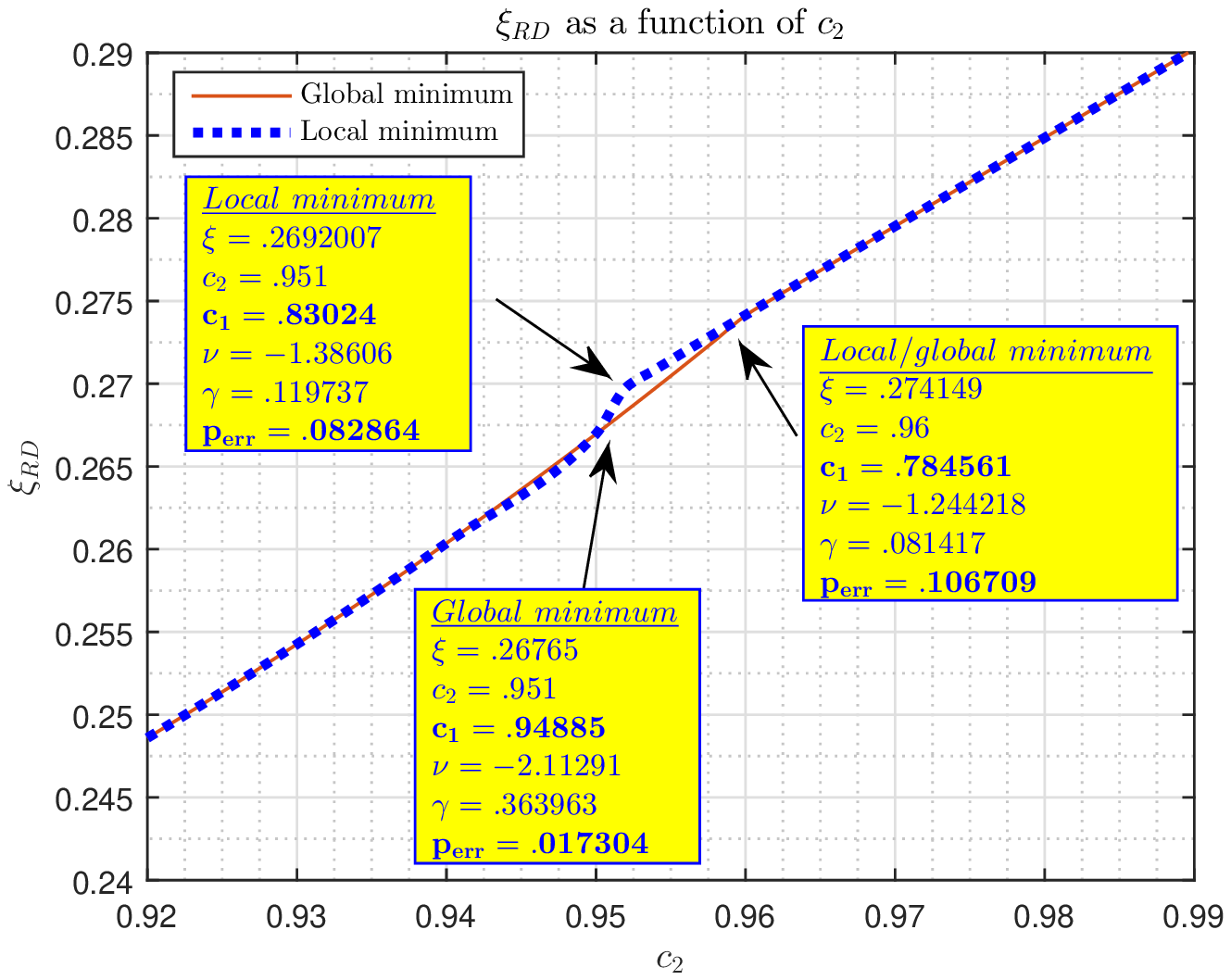,width=8cm,height=6cm}}
\end{minipage}
\vspace{-0.3in}
\caption{$\xi_{RD}$ as a function of $c_2$; $1/\sigma^2=9$ db a) left - full range of $c_2$, b) right - appearance of local optima}
\label{fig:figdisc12a}
\end{figure}

\subsubsection{CLuP -- avoiding lower stationary point}
\label{sec:discussionclupavoidloerstatpts}

As earlier calculations and Figures \ref{fig:figdisc12} and \ref{fig:figdisc12a} indicate there are clearly two stationary points that might be of interest when looking at the CLuP's performance. One is obviously interested only in the one that is to the right in both figures. That immediately of course raises the question as to how the CLuP performs when it comes to avoiding the so-called lower stationary point. The key to understanding that is the CLuP's first step (iteration) which amounts to determining $\x^{(1)}$ as
\begin{eqnarray}
\x^{(1)}=\frac{\x^{(1,s)}}{\|\x^{(1,s)}\|_2} \quad \mbox{with}\quad \x^{(1,s)}=\mbox{arg}\min_{\x} & & -(\x^{(0)})^T\x  \nonumber \\
\mbox{subject to} & & \|\y-A\x\|\leq r\nonumber \\
&& \x\in \left [-\frac{1}{\sqrt{n}},\frac{1}{\sqrt{n}}\right ]^n. \label{eq:avoidclup1}
\end{eqnarray}
This is exactly the same problem that we considered in great detail in \cite{Stojnicclupcmpl19}. Moreover, after a bit of juggling one arrives at its a more relevant version
\begin{eqnarray}
\xi_{p,1}(\alpha,\sigma,c_{1,z},s_1)=\lim_{n\rightarrow\infty}\frac{1}{\sqrt{n}}\mE \min_{\z} & & \|\sigma\v+A\z\|_2  \nonumber \\
\mbox{subject to} & & \|\z\|_2^2=c_{1,z}\nonumber \\
& & (\x^{(0)})^T\z=s_1 \nonumber \\
&& \z\in \left [0,2/\sqrt{n}\right ]^n. \label{eq:avoidclup1a3}
\end{eqnarray}
The following theorem is proven in \cite{Stojnicclupcmpl19}.
\begin{theorem}(CLuP -- RDT estimate -- first iteration \cite{Stojnicclupcmpl19}) Set
\begin{equation}
I_{box,1}^{(1)}(\gamma,\nu)=\rho I_{1,1}(\gamma,\nu)+(1-\rho) I_{1,1}(\gamma,-\nu)+\rho I_{2,1}(\gamma,\nu)+(1-\rho) I_{2,1}(\gamma,-\nu),\label{eq:avoidclup15}
\end{equation}
where
\begin{eqnarray}
I_{1,1}(\gamma,\nu) &  = & -(exp(-0.5(4 \gamma + \nu)^2) (\nu - 4 \gamma) + \sqrt{\pi/2} (\nu^2 + 1) \erf(2\sqrt{2}\gamma + 1/\sqrt{2}\nu) \nonumber \\
& & - \sqrt{\pi/2} (\nu^2 + 1) \erf(\nu/\sqrt{2}) - exp(-0.5 \nu^2) \nu)/(4 \sqrt{2 \pi}\gamma) \nonumber \\
I_{2,1}(\gamma,\nu)  &  = & (4\gamma+2\nu).5\erfc((4\gamma+\nu)/\sqrt{2})-2exp(-1/2(4\gamma+\nu)^2)/\sqrt{2\pi}.
\label{eq:avoidclup16}
\end{eqnarray}
Moreover, set
\begin{equation}
\xi_{RD}^{(1)}(\alpha,\sigma;c_{1,z},s_1,\gamma,\nu)=\sqrt{\alpha}\sqrt{c_{1,z}+\sigma^2}+I_{box,1}^{(1)}(\gamma,\nu)-\nu s_1-\gamma c_{1,z}. \label{eq:avoidclup17}
\end{equation}
Let $\xi_{p,1}(\alpha,\sigma,c_{1,z},s_1)$ be as in (\ref{eq:avoidclup1a3}). Then \begin{equation}
\xi_{p,1}(\alpha,\sigma,c_{1,z},s_1)= \max_{\gamma,\nu}\xi_{RD}^{(1)}(\alpha,\sigma;c_{1,z},s_1,\gamma,\nu).\label{eq:thmavoidcluprd1}
\end{equation}
Consequently,
\begin{equation}
\min_{c_{1,z}}\xi_{p,1}(\alpha,\sigma,c_{1,z},s_1)= \min_{c_{1,z}}\max_{\gamma,\nu}\xi_{RD}^{(1)}(\alpha,\sigma;c_{1,z},s_1,\gamma,\nu).
\label{eq:thmcluprd2}
\end{equation}\label{thm:avoidcluprd1}
\end{theorem}
\begin{proof}
Follows from the general RDT concepts presented in \cite{StojnicCSetam09,StojnicISIT2010binary,StojnicDiscPercp13,StojnicGenLasso10,StojnicGenSocp10,StojnicPrDepSocp10,StojnicRegRndDlt10}, the discussion presented in \cite{Stojnicclupcmpl19}, and the fact that the \bl{\textbf{strong random duality}} trivially holds.
\end{proof}
The analysis in \cite{Stojnicclupcmpl19} proceeds further and by utilizing the strong random duality characterizes the exact estimates for all other quantities that may be of interest. To do so it first solves the following optimization problem
\begin{eqnarray}
\{\hat{\nu}^{(1)},\hat{\gamma}^{(1)},\hat{c}_{1,z}^{(1)},\hat{s}_1^{(1)}\}=\mbox{arg} \min_{s_1} & & s_1 \nonumber \\
\mbox{subject to} & & \min_{0\leq c_{1,z}\leq 4}\max_{\gamma,\nu}\xi_{RD}^{(1)}(\alpha,\sigma;c_{1,z},s_1,\gamma,\nu),\label{eq:avoidclup17a}
\end{eqnarray}
and then defines
\begin{eqnarray}
s_{x,1}(\gamma,\nu) & = & -\nu/2/\gamma(.5\erfc(\nu/\sqrt{2})-.5\erfc((\nu+4\gamma)/\sqrt{2}))\nonumber \\
& & +1/2/\gamma/\sqrt{2\pi}(exp(-\nu^2/2)-exp(-(4\gamma+\nu)^2/2))\nonumber \\
s_{xsq,1}(\gamma,\nu) & = & -I_{1,1}(\gamma,\nu)/\gamma\nonumber \\
s_{x,2}(\gamma,\nu) & = & 2(.5\erfc((4\gamma+\nu)/\sqrt{2}))\nonumber \\
s_{xsq,2}(\gamma,\nu) & = & 2s_{x,2},\label{eq:avoidclup18}
\end{eqnarray}
to finally obtain
\begin{eqnarray}
\mE((\x_{sol})^T\x) & = & 1-(\rho s_{x,1}(\hat{\gamma}^{(1)},\hat{\nu}^{(1)})+(1-\rho)s_{x,1}(\hat{\gamma}^{(1)},-\hat{\nu}^{(1)})+\rho s_{x,2}(\hat{\gamma}^{(1)},\hat{\nu}^{(1)})+(1-\rho)s_{x,2}(\hat{\gamma}^{(1)},-\hat{\nu}^{(1)}))\nonumber \\
\mE\|\x\|_2^2 & = & \rho s_{xsq,1}(\hat{\gamma}^{(1)},\hat{\nu}^{(1)})+(1-\rho)s_{xsq,1}(\hat{\gamma}^{(1)},-\hat{\nu}^{(1)})+\rho s_{xsq,2}(\hat{\gamma}^{(1)},\hat{\nu}^{(1)})+(1-\rho)s_{xsq,2}(\hat{\gamma}^{(1)},-\hat{\nu}^{(1)})\nonumber \\
&&+2\mE((\x_{sol})^T\x)-1.\label{eq:avoidclup21}
\end{eqnarray}
 Moreover, \cite{Stojnicclupcmpl19} also gets the estimate for the probability of error
\begin{equation}\label{eq:avoidclup21a}
  p_{err,1}=1-P\left (\z_i\leq \frac{1}{\sqrt{n}}\right )=1-\left (\rho\left (\frac{1}{2}\erfc\left ( \frac{-2\hat{\gamma}^{(1)}-\hat{\nu}^{(1)}}{\sqrt{2}}\right )\right )+(1-\rho)\left ( \frac{1}{2}\erfc\left ( \frac{-2\hat{\gamma}^{(1)}+\hat{\nu}^{(1)}}{\sqrt{2}}\right )\right )\right ).
\end{equation}
The theoretical values obtained based on Theorem \ref{thm:cluprd1} in \cite{Stojnicclupcmpl19} for various system parameters are shown in Table \ref{tab:tabavoidnum1} for two different values of SNR, $1/\sigma^2=10$[db] and $1/\sigma^2=13$[db].
\begin{table}[h]
\caption{\textbf{Theoretical} values for various system parameters obtained based on Theorem \ref{thm:avoidcluprd1}}\vspace{.1in}
\hspace{-0in}\centering
\small{
\begin{tabular}{||c||c|c||c|c||c|c|c|c||}\hline\hline
$1/\sigma^2 $[db]  & $\hat{\nu}$ & $\hat{\gamma}$ & $\hat{c}_{1,z}$ & $\hat{s}_1$ &   $\xi_{RD}^{(1)} $ & $p_{err,1} $  & $\|\x\|_2^2$ &
$(\x_{sol})^T\x$ \\ \hline\hline
$10  $ & $\mathbf{0.5075}  $ & $\mathbf{0.6816}  $ & $\mathbf{0.3306}  $ & $\mathbf{-0.1844}  $ & $\mathbf{0.2252}  $ & $\mathbf{0.1134}  $ & $\mathbf{0.6749}  $ & $\mathbf{0.6722}  $ \\ \hline
$13  $ & $\mathbf{0.4953}  $ & $\mathbf{0.9420}  $ & $\mathbf{0.1753}  $ & $\mathbf{-0.1314}  $ & $\mathbf{0.1594}  $ & $\mathbf{0.0456}  $ & $\mathbf{0.7009}  $ & $\mathbf{0.7628}  $ \\ \hline\hline
\end{tabular}}
\label{tab:tabavoidnum1}
\end{table}
 As the level of precision that the \bl{\textbf{Random duality theory}} achieves is often very impressive even for moderate problem dimensions the corresponding simulated values are shown in Table \ref{tab:tabavoidnum2}. We choose $\alpha=0.8$ and $n=400$.
\begin{table}[h]
\caption{\textbf{Theoretical}/\bl{\textbf{simulated}} values for various system parameters obtained based on Theorem \ref{thm:avoidcluprd1}}\vspace{.1in}
\hspace{-0in}\centering
\small{
\begin{tabular}{||c||c||c|c|c|c||}\hline\hline
$1/\sigma^2 $[db]  &  $\hat{s}_1$ &   $\xi_{RD}^{(1)} $ & $p_{err,1} $  & $\|\x\|_2^2$ &
$(\x_{sol})^T\x$ \\ \hline\hline
$10  $ & $\mathbf{-0.1844  }/\bl{\mathbf{0.1845  }}$ & $\mathbf{-0.2252  }/\bl{\mathbf{0.2252  }}$ & $\mathbf{0.1134  }/\bl{\mathbf{0.1133  }}$ & $\mathbf{0.6749  }/\bl{\mathbf{0.6769  }}$ & $\mathbf{0.6722  }/\bl{\mathbf{0.6719  }}$ \\ \hline
$13  $ & $\mathbf{-0.1314  }/\bl{\mathbf{0.1316  }}$ & $\mathbf{-0.1594  }/\bl{\mathbf{0.1594  }}$ & $\mathbf{0.0456  }/\bl{\mathbf{0.0483  }}$ & $\mathbf{0.7009  }/\bl{\mathbf{0.7005  }}$ & $\mathbf{0.7628  }/\bl{\mathbf{0.7596  }}$ \\ \hline\hline
\end{tabular}}
\label{tab:tabavoidnum2}
\end{table}
It is relatively easy to observe a very strong agreement between the theoretical and simulated values. Also, as the value in the table indicates, one has $c_2=0.6749$ (or when it comes to the simulated value $c_2=0.6769$) which is well above $0.46075$ given in (\ref{eq:discclup17b}). Given that the CLuP's objective trivially never decreases one then indeed easily has that the lower stationary point will be circumvented. This is rather clear from Figure \ref{fig:figdisc12} as well.

\subsubsection{Moving from $0$FL RDT to $1$FL RDT}
\label{sec:discussionclupmovingfrom0FLMLto1FLML}

As we have mentioned earlier, the results that we presented utilizing RDT for ML are expected to need some corrections in the low SNR regimes. We earlier showed a set of results that one can obtain utilizing the so-called 1FL RDT. They were related to $p_{err}$. In Figure \ref{fig:figdisc13} we show an analogous set of results for $\xi$. As can be seen the value of the objective $\xi$ is substantially lifted through the 1FL RDT mechanism. More importantly the relatively low value of $c_1=.7195$ where one achieves the $\xi$ minimum for 0FL RDT is now replaced by a significantly larger one $c_1=.955$. Correspondingly, one has a substantially lower estimate for probability of error as already shown in Figure \ref{fig:figdisc1}. Now, in a similar fashion one can redo the whole 1FL RDT mechanism for CLuP as well and then reanalyze all of the above functions and behavior (reemergence/disappearance) of their potential local/global optima. We will address that in one of the companion papers. However, we do mention here that in those scenarios one does not have the type of changes that we have here for the ML.
\begin{figure}[htb]
\centering
\centerline{\epsfig{figure=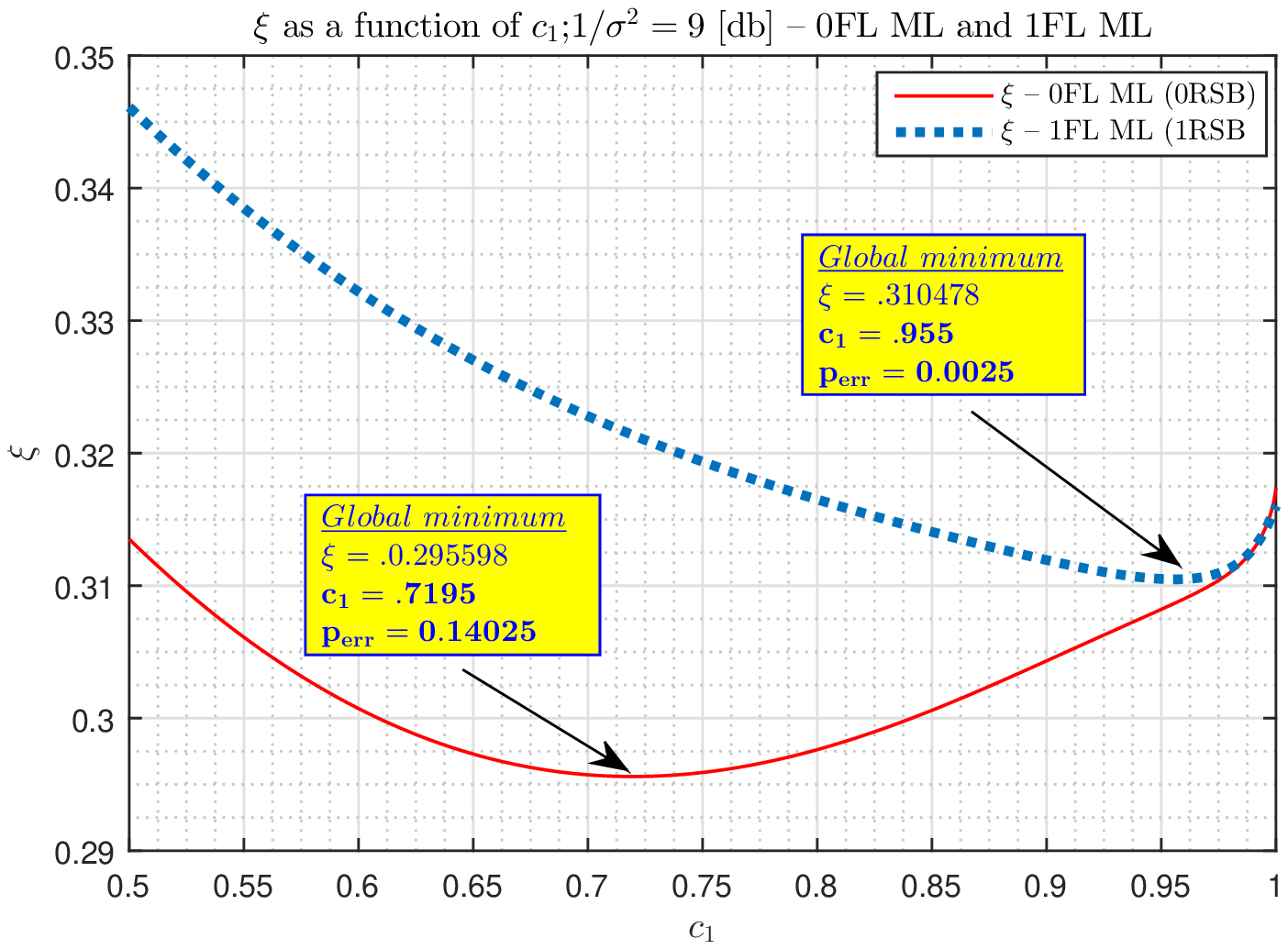,width=11.5cm,height=8cm}}
\caption{$\xi_{RD}$ as a function of $c_1$; $\alpha=0.8$; $1/\sigma^2=9$; -- $0$FL ML and $1$FL ML RDT}
\label{fig:figdisc13}
\end{figure}

\section{Numerical simulations}
\label{sec:numres}

In this section we present some of the numerical results to complement the above theoretical considerations.

\subsection{ML -- numerical experiments}
\label{sec:numresml}

We start with the ML performance. Since the original problem (\ref{eq:ml1}) is obviously hard we implemented a fast bit-flipping algorithmic heuristic to solve it. While there is no guarantee that the solutions that we have found are optimal the results presented in Figure \ref{fig:fignum1} indicate that they may actually be very close to the optimal ones. We should also add that despite their excellent performance the analysis of these types of algorithms remains a challenge. We also complement Figure \ref{fig:fignum1} with Table \ref{tab:tabnum1} where some of the numerical values and parameters of the simulated systems are given as well. Although it should be clear by itself, we add that $0$FLML and $1$FLML in superscripts denote estimates obtained based on 0FL and 1FL RDT. As one can see from both, Figure \ref{fig:fignum1} and Table \ref{tab:tabnum1}, all of the above considerations seem to be in a very strong agreement with the results obtained through numerical experiments (the tiny differences that still remain for lowest $1/\sigma^2$ would virtually disappear on the second level of RDT, 2FL RDT).
\begin{figure}[htb]
\centering
\centerline{\epsfig{figure=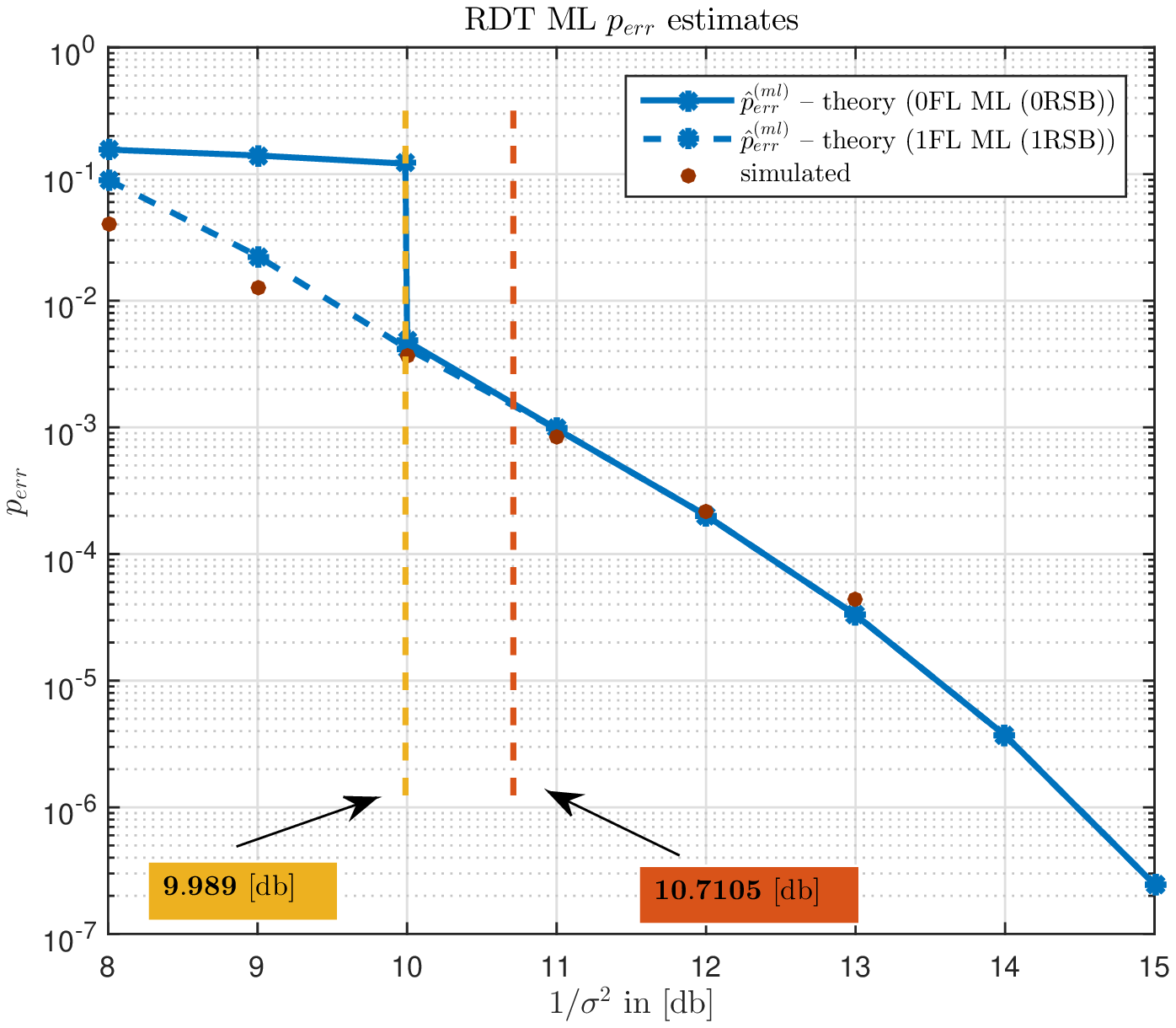,width=11.5cm,height=8cm}}
\caption{$p_{err}^{(ml)}$ as a function of $1/\sigma^2$; $\alpha=0.8$ -- theory and simulations}
\label{fig:fignum1}
\end{figure}

\begin{table}[h]
\caption{\textbf{Theoretical}/\bl{\textbf{simulated}} values for $\xi$ and $p_{err}$; the $p_{err}$ values correspond to the data in Figure \ref{fig:fignum1}}\vspace{.1in}
\hspace{-0in}\centering
\small{
\begin{tabular}{||c||c|c|c||c|c|c||}\hline\hline
$1/\sigma^2 $[db]  & $\xi_{RD}^{(0FLML)} $ & $\xi_{RD}^{(1FLML)} $ & $\xi$--simulated & $p_{err}^{(0FLML)} $ & $p_{err}^{(1FLML)} $ & $p_{err}$--simulated \\ \hline\hline
$8  $ & $\mathbf{3.1259e-01}  $ & $\mathbf{3.3339e-01}  $ & $\bl{\mathbf{3.3854e-01}}  $ & $\mathbf{1.56e-01}  $ & $\mathbf{9.00e-02}  $ & $\bl{\mathbf{4.01e-02}}  $ \\ \hline
$9  $ & $\mathbf{2.9560e-01}  $ & $\mathbf{3.1048e-01}  $ & $\bl{\mathbf{3.1107e-01}}  $ & $\mathbf{1.40e-01}  $ & $\mathbf{2.25e-02}  $ & $\bl{\mathbf{1.29e-02}}  $ \\ \hline
$10  $ & $\mathbf{2.8092e-01}  $ & $\mathbf{2.8099e-01}  $ & $\bl{\mathbf{2.8061e-01}}  $ & $\mathbf{4.77e-03} $ & $\mathbf{4.20e-03}  $ & $\bl{\mathbf{3.74e-03}}  $ \\ \hline
$11  $ & $\mathbf{2.5162e-01}  $ & $\mathbf{2.5162e-01}  $ & $\bl{\mathbf{2.5030e-01}}  $ & $\mathbf{9.72e-04}  $ & $\mathbf{9.72e-04}  $ & $\bl{\mathbf{8.43e-04}}  $ \\ \hline
$12  $ & $\mathbf{2.2457e-01}  $ & $\mathbf{2.2457e-01}  $ & $\bl{\mathbf{2.2447e-01}}  $ & $\mathbf{2.01e-04}  $ & $\mathbf{2.01e-04}  $ & $\bl{\mathbf{2.14e-04}}  $ \\ \hline
$13  $ & $\mathbf{2.0022e-01}  $ & $\mathbf{2.0022e-01}  $ & $\bl{\mathbf{1.9994e-02}}  $ & $\mathbf{3.30e-05}  $ & $\mathbf{3.30e-05}  $ & $\bl{\mathbf{4.36e-05}}  $ \\ \hline\hline
\end{tabular}}
\label{tab:tabnum1}
\end{table}

\subsection{CLuP -- numerical experiments}
\label{sec:numresclup}

We now switch to the CLuP's performance. In Figure \ref{fig:fignum2} we present results obtained from numerical experiments for all the three choices of $r_{sc}$ that we considered earlier, i.e. for $r_{sc}=\{1.1,1.3,1.5\}$. We mostly focus on the SNR regime above the first line of corrections where, based on the above analysis, one is to expect a good performance. The results for $r_{sc}=\{1.1,1.3\}$ were obtained using $n=400$. To get a bit better concentrations closer to the ML for $r_{sc}=1.5$ we used $n=800$. We complement these results with the numerical values in Tables \ref{tab:tabnum2}, \ref{tab:tabnum3}, and \ref{tab:tabnum4}. In addition to the probabilities of error we in tables present results for two key CLuP parameters $c_2$ and $c_1$ as well. We again observe an excellent agreement between the theoretical predictions and the results obtained through numerical experiments. In particular, already for rather moderately small value $n=800$ the results are almost identical to the theoretical predictions.
\begin{figure}[htb]
\centering
\centerline{\epsfig{figure=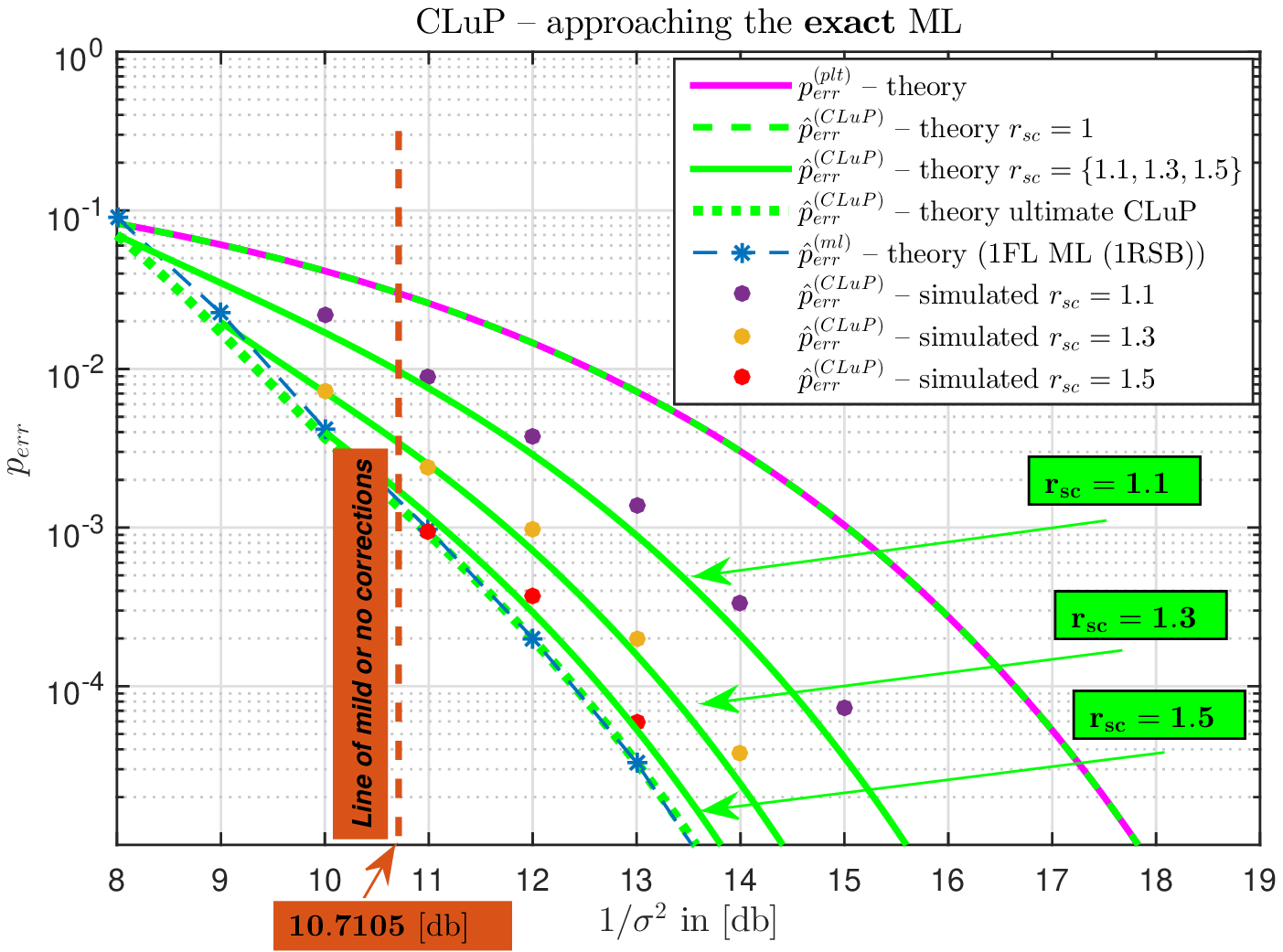,width=11.5cm,height=8cm}}
\caption{$p_{err}$ as a function of $1/\sigma^2$; $\alpha=0.8$ -- theory and simulations}
\label{fig:fignum2}
\end{figure}

\begin{table}[h]
\caption{\textbf{Theoretical}/\bl{\textbf{simulated}} values for $c_2$, $c_1$, and $p_{err}^{(CLuP)} $; the $p_{err}^{(CLuP)} $ values correspond to the data in Figure \ref{fig:fignum2}; $\alpha=0.8$; $r_{sc}=1.1$; $n=400$}\vspace{.1in}
\hspace{-0in}\centering
\small{
\begin{tabular}{||c||c|c||c|c||c|c||}\hline\hline
$1/\sigma^2 $[db]  & $c_2 $ & $c_2$--simulated  & $c_1$ & $c_1$--simulated & $p_{err}^{(CLuP)} $ & $p_{err}^{(CLuP)} $--simulated \\ \hline\hline
$10  $ & $\mathbf{8.420e-01}  $ & $\bl{\mathbf{8.370e-01}}  $ & $\mathbf{8.820e-01}  $ & $\bl{\mathbf{8.730e-01}}  $ & $\mathbf{1.698e-02}  $ & $\bl{\mathbf{2.177e-02}} $ \\ \hline\hline
$11  $ & $\mathbf{8.520e-01}  $ & $\bl{\mathbf{8.453e-01}}  $ & $\mathbf{8.980e-01}  $ & $\bl{\mathbf{8.922e-01}}  $ & $\mathbf{7.559e-03} $ & $\bl{\mathbf{8.880e-03}}  $ \\ \hline
$12  $ & $\mathbf{8.628e-01}  $ & $\bl{\mathbf{8.600e-01}}  $ & $\mathbf{9.105e-01}  $ & $\bl{\mathbf{9.080e-01}}  $ & $\mathbf{2.886e-03}  $ & $\bl{\mathbf{3.727e-03}}  $ \\ \hline
$13  $ & $\mathbf{8.738e-01}  $ & $\bl{\mathbf{8.700e-01}}  $ & $\mathbf{9.210e-01}  $ & $\bl{\mathbf{9.180e-01}}  $ & $\mathbf{8.922e-04}  $ & $\bl{\mathbf{1.380e-03}}  $ \\ \hline
$14  $ & $\mathbf{8.845e-01}  $ & $\bl{\mathbf{8.818e-01}}  $ & $\mathbf{9.300e-01}  $ & $\bl{\mathbf{9.279e-01}}  $ & $\mathbf{2.106e-04}  $ & $\bl{\mathbf{3.341e-04}}  $ \\ \hline
$15  $ & $\mathbf{8.945e-01}  $ & $\bl{\mathbf{8.930e-01}}  $ & $\mathbf{9.377e-01}  $ & $\bl{\mathbf{9.365e-01}}  $ & $\mathbf{3.554e-05}  $ & $\bl{\mathbf{7.275e-05}}  $ \\ \hline\hline
\end{tabular}}
\label{tab:tabnum2}
\end{table}

\begin{table}[h]
\caption{\textbf{Theoretical}/\bl{\textbf{simulated}} values for $c_2$, $c_1$, and $p_{err}^{(CLuP)} $; the $p_{err}^{(CLuP)} $ values correspond to the data in Figure \ref{fig:fignum2}; $\alpha=0.8$; $r_{sc}=1.3$; $n=400$}\vspace{.1in}
\hspace{-0in}\centering
\small{
\begin{tabular}{||c||c|c||c|c||c|c||}\hline\hline
$1/\sigma^2 $[db]  & $c_2 $ & $c_2$--simulated  & $c_1$ & $c_1$--simulated & $p_{err}^{(CLuP)} $ & $p_{err}^{(CLuP)} $--simulated \\ \hline\hline
$11  $ & $\mathbf{9.350e-01}  $ & $\bl{\mathbf{9.320e-01}}  $ & $\mathbf{9.565e-01}  $ & $\bl{\mathbf{9.540e-01}}  $ & $\mathbf{2.487e-03}  $ & $\bl{\mathbf{2.421e-03}}  $ \\ \hline
$12  $ & $\mathbf{9.400e-01}  $ & $\bl{\mathbf{9.371e-01}}  $ & $\mathbf{9.622e-01}  $ & $\bl{\mathbf{9.602e-01}}  $ & $\mathbf{7.177e-04}  $ & $\bl{\mathbf{9.804e-04}}  $ \\ \hline
$13  $ & $\mathbf{9.451e-01}  $ & $\bl{\mathbf{9.432e-01}}  $ & $\mathbf{9.668e-01}  $ & $\bl{\mathbf{9.656e-01}}  $ & $\mathbf{1.575e-04}  $ & $\bl{\mathbf{1.996e-04}}  $ \\ \hline
$14  $ & $\mathbf{9.500e-01}  $ & $\bl{\mathbf{9.489e-01}}  $ & $\mathbf{9.707e-01}  $ & $\bl{\mathbf{9.700e-01}}  $ & $\mathbf{2.422e-05}  $ & $\bl{\mathbf{3.748e-05}}  $ \\ \hline\hline
\end{tabular}}
\label{tab:tabnum3}
\end{table}

\begin{table}[h]
\caption{\textbf{Theoretical}/\bl{\textbf{simulated}} values for $c_2$, $c_1$, and $p_{err}^{(CLuP)} $; the $p_{err}^{(CLuP)} $ values correspond to the data in Figure \ref{fig:fignum2}; $\alpha=0.8$; $r_{sc}=1.5$; $n=800$}\vspace{.1in}
\hspace{-0in}\centering
\small{
\begin{tabular}{||c||c|c||c|c||c|c||}\hline\hline
$1/\sigma^2 $[db]  & $c_2 $ & $c_2$--simulated  & $c_1$ & $c_1$--simulated & $p_{err}^{(CLuP)} $ & $p_{err}^{(CLuP)} $--simulated \\ \hline\hline
$11  $ & $\mathbf{9.815e-01}  $ & $\bl{\mathbf{9.805e-01}}  $ & $\mathbf{9.872e-01}  $ & $\bl{\mathbf{9.860e-01}}  $ & $\mathbf{1.187e-03}  $ & $\bl{\mathbf{9.375e-04}}  $ \\ \hline
$12  $ & $\mathbf{9.829e-01}  $ & $\bl{\mathbf{9.828e-01}}  $ & $\mathbf{9.892e-01}  $ & $\bl{\mathbf{9.889e-01}}  $ & $\mathbf{2.926e-04}  $ & $\bl{\mathbf{3.750e-04}}  $ \\ \hline
$13  $ & $\mathbf{9.843e-01}  $ & $\bl{\mathbf{9.843e-01}}  $ & $\mathbf{9.906e-01}  $ & $\bl{\mathbf{9.905e-01}}  $ & $\mathbf{5.334e-05}  $ & $\bl{\mathbf{6.000e-05}}  $ \\ \hline\hline
\end{tabular}}
\label{tab:tabnum4}
\end{table}

\section{Summary}
\label{sec:summary}

As the mechanisms that we presented in previous sections are a somewhat complex interplay of many factors we in this section provide a brief summary of the key points. However, as we have mentioned on multiple occasions, this is the introductory paper and we try to stay away from technical complications as much as possible. 

To start things off we in Figure \ref{fig:figsummary1} give the summary of the main performance feature discussed in the previous sections. That feature is of course the probability of error, $p_{err}$, and the figure itself is of course the highlighting Figure \ref{fig:fighighlightclup}. We first have the 1FLML and the ultimate CLuP's calculated performance curves that essentially characterize attacking the ML problem on the so-called \textbf{exact} level. We also plot the standard polynomial heuristics based on the (convex) relaxations of the given discrete ${\cal X}$ set. These include, the ball, polytope, and the SDP heuristics. As mentioned earlier, as these are convex problems their performance is relatively easy to derive through the \bl{\textbf{random duality}}. In fact, we have seen earlier that the polytope one is essentially trivial and a direct consequence of many results that we have already created, most notably those from \cite{StojnicDiscPercp13,StojnicGenLasso10,StojnicGenSocp10,StojnicPrDepSocp10,StojnicRegRndDlt10}. The analysis of the ball relaxation is even more trivial and we show the plot without even bothering to explain all these trivialities. The SDP is also relatively easy to derive, however the final results are a bit more involved and we will present them in a separate paper. All these three heuristics were introduced essentially as first steps in the branch-and-bound mechanism designed in \cite{StojnicBBSD05,StojnicBBSD08}. At that time it was observed that they substantially trail the designed branch-and-bound mechanism and on their own are essentially of no use when it comes to solving the problem exactly. With the appearance of the random duality theory these observations are also very simple to precisely mathematically characterize. As the random duality based theoretical characterizations in Figure \ref{fig:figsummary1} confirm, all three of these heuristics indeed substantially trail the ML and CLuP results. In fact, it is actually a known thing that as $\alpha$ gets smaller the failure of typical convexity type of techniques gets more pronounced. This is not necessarily particularly true only for the problem at hand but for many similar ones as well. Basically, as $\alpha$ gets smaller the problem becomes much harder and for $\alpha\rightarrow 0$ it becomes one of the hardest well-known optimization problems where hardly any known technique can succeed and convexity based ones dramatically fail. Moreover, as $\alpha$ gets smaller even CLuP can occasionally experience difficulties. However, there are ways to remedy that. They are based on designing a bit more sophisticated CLUP's variants that we will discuss in separate papers. We do however mention right here that the performance gain over the standard convex techniques becomes even more substantial as $\alpha$ gets smaller.
\begin{figure}[htb]
\centering
\centerline{\epsfig{figure=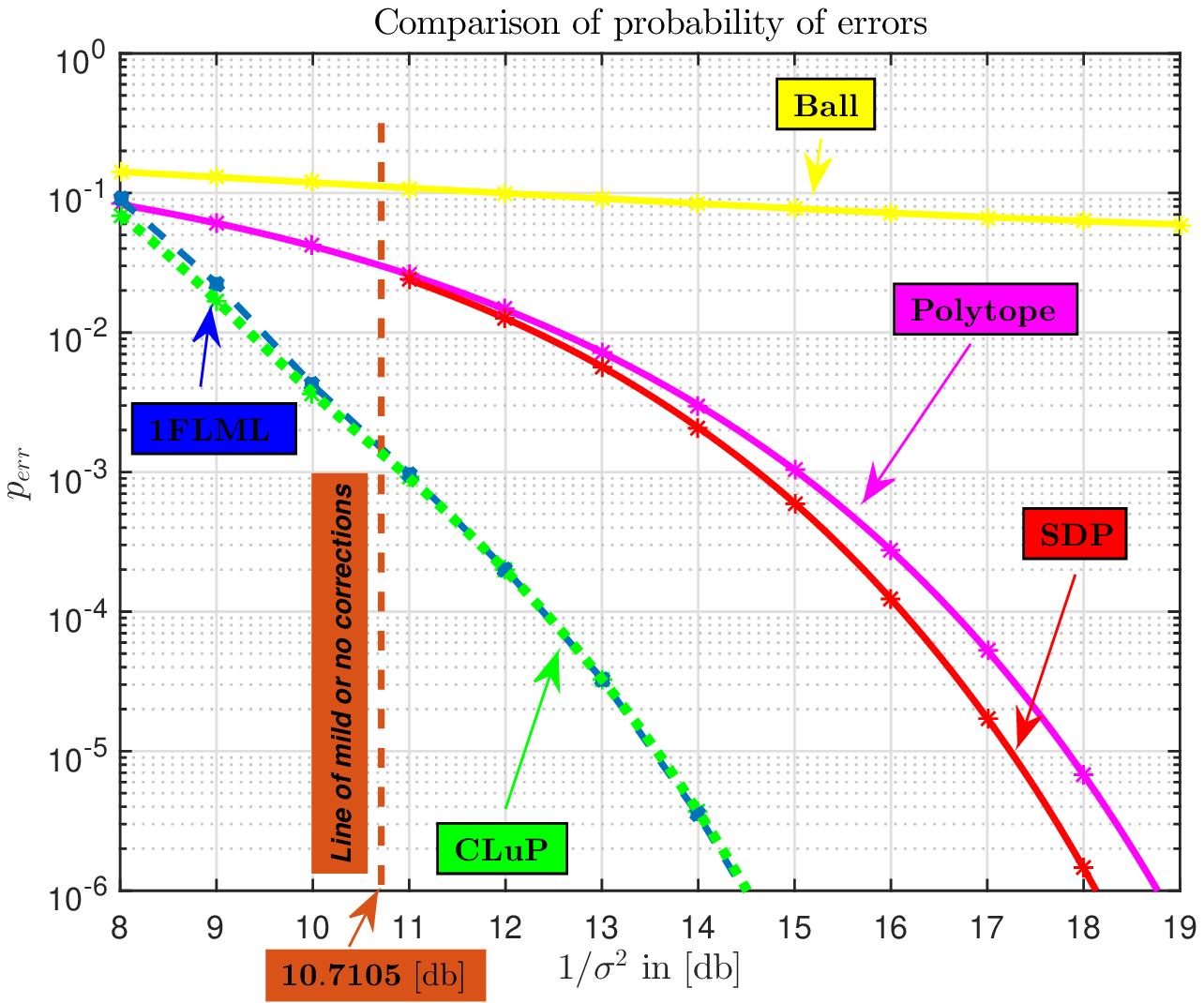,width=11.5cm,height=8cm}}
\caption{Comparison of $p_{err}$ as a function of $1/\sigma^2$; $\alpha=0.8$}
\label{fig:figsummary1}
\end{figure}

Another thing that was clear from the above considerations is that there is a $\sigma$-dependent breaking point where one may need to modify the original CLuP setup. There are many ways how this can be done and we will consider both, simple and highly advances modifications in separate papers. As this is the introductory paper we insist on the simplest possible structure. In passing we just mention that a couple of simple modifications typically useful in the lower SNR regimes include restarting the algorithms a few times with a different $\x_0$ as well as constraining additionally with $\hat{\x}^T\x\geq \hat{c}_1$ to avoid local minima over $c_1$, where $\hat{\x}$ and $\hat{c}_1$ are estimates for $\x$ and $c_1$. These can be obtained in various ways; one of them, for example, would be to utilize one of the above convex relaxation heuristics, say the polytope one. In particular, one can run say $j$ times the following
\begin{equation*}
[\bl{\x^{(i)},\x^{(CLuP,j)}}]=[\bl{\mbox{CLuP}(\y,A,r,\x^{(0,j)},\emptyset,i_{max},\delta)}]
\end{equation*}
for $j$ different $\x^{(0,j)}$ generated either randomly or in specific way and then choose the one that converges to the highest objective, i.e. produces the largest $\|\x^{(CLuP,j)}\|_2^2$. Or if one wants to be a bit more specific about avoiding particular $c_1$ local optima, one can first obtain $\x^{(0)}$ through a convex heuristic. For simplicity, say one again chooses the polytope one, i.e. sets $\x^{(0)}=\x_{plt}$ and then runs
\begin{eqnarray*}
[\bl{\x^{(i)},\x^{(CLuP)}}] & = & [\bl{\mbox{CLuP}(\y,A,r,\x_{plt},\{\x_{plt}^T\x\geq c_1^{(plt)}\},i_{max},\delta)}],
\end{eqnarray*}
(where the estimate for $c_1^{(plt)}$ is obtained through the above RDT polytope relaxation characterization), to obtain a good starting point $\x^{(CLuP)}$ that can potentially help avoiding the local $c_1$ optima in the second running of the standard CLuP
\begin{eqnarray*}
[\bl{\x^{(i)},\x^{(CLuP)}}] & = & [\bl{\mbox{CLuP}(\y,A,r,\x^{(CLuP)},\emptyset,i_{max},\delta)}].
\end{eqnarray*}

In fact, in Figure \ref{fig:fignum2} for $r_{sc}=1.3$ and $1/\sigma^2=10$[db] we have implemented this strategy as there were around $10\%$ instances where the CLuP's performance wouldn't be as expected. Since $1/\sigma^2=10$[db] is in the regime below one of the critical lines this is in a way to be expected, if for no other reason then at least because the dimensions are finite and it may happen that one runs into bad instances where big dimensions may be needed for everything to kick in (of course, the other reasons that are way more likely to cause the problems we have discussed earlier). However, with these modifications the performance got back to match exactly what the theory predicts. We also mention that for $r_{sc}=1.1$ we didn't find that this type of modification was needed for $1/\sigma^2=10$[db] but it was needed for $1/\sigma^2=9$[db]. Still, it certainly wouldn't hurt to utilize it anyway.

Another important thing that we haven't touched upon until now is of course the overall complexity of the algorithm. The reason of course is that we will  have a whole lot more to say on this topic and it will in fact be the key topic in several of our companion papers. Here we just briefly mention that in the favorable regime (above the first line of corrections) the number of iterations needed for algorithm to get to a $10^{-8}$ level of convergence (the objective difference between two successive iterates) was rarely over $20$. However, this is a huge overestimate, as the typical number might in some scenarios be even less than $10$. We should also emphasize that this is actually independent of $n$ and it depends almost exclusively on $r_{sc}$ and $\sigma$. This of course ultimately means that overall complexity is basically matching the complexity of solving a quadratic program which is clearly polynomial.

Also, we should add that here we consider the simplest possible version of the algorithm. As we have just discussed above, for example, instead of starting with $\x_0$ that is randomly generated one can choose it as a solution of one of the convex/polynomial heuristics. We will also analyze these scenarios in one of our companion papers in details. Here we just briefly mention that choosing carefully the starting $\x_0$ can be beneficial for both, handling the hard regimes below the lines of corrections as mentioned above but also for lowering the complexity (basically the number of running iterations).

Various other options are possible as well. For example, one can choose a way more sophisticated iterative updating strategies that include further modifying the objective or the constraints set. Moreover, one can also do multiple runs of CLuP with different radius. A particularly successful strategy that we have found is to successively increase the radius until one reaches the level close to ML. This type of strategy can then also be combined with all the other ones that we have already mentioned. Another very interesting option is to successively change the radius through the iterations within a single running of CLuP. In other words, instead of keeping $r$ fixed, one can have $r_i$. Such a modification can substantially improve even a single running of CLuP in any regime. A relevant technical problem that we looked at is how to carefully design the sequence $r_i$ so that the complexity is minimal and the overall performance optimal. A substantial improvement can be achieved though such a consideration as well. As we have stated on numerous occasions, since this is the introductory paper we navigated the presentation accordingly and tried to stick with the simplest possible structure of the algorithm. Obviously, we designed a way more advanced ones and we will discuss them in separate companion papers.

Finally, as it is probably obvious from the entire presentation, the mechanism that we presented in this paper is in no way restricted to the MIMO ML problem discussed here. We selected this problem to be the one where we will showcase the concepts due to its enormous importance/relevance in many scientific fields, starting with information theory and signal processing, then moving to statistics, machine learning, and many others. We have already applied it in all of these fields on a very large collection of problems. All of the above discussion and summarizing points apply to all of these problems as well and quite a few additional features emerge due to problems specifics in various different fields. We will systematically present all of these results in a large collection of companion papers.

\section{Conclusion}
\label{sec:conc}

In this paper we introduce a simple yet very powerful  concept for achieving in polynomial time the \textbf{exact} ML performance in MIMO systems. We refer to the concept as the Controlled Loosening-up (CLuP). It turns out that CLuP performs remarkably well. In particular, not only does it achieve an excellent performance in terms of accuracy, it actually does so rather quickly with a very small number of iterations. In fact, not only can an excellent performance be achieved through a number of iterations that is polynomial but actually a very small \textbf{fixed} number (basically independent of the problem dimensions) of iterations suffices as well.

While the structure of the algorithm is very simple and the performance is excellent, the rationale and technical foundation behind all of it do require a very careful discussion. We provided some of the key points that give a general picture as to how/why the entire mechanism actually functions. In particular, we observed that it is naturally connected to the ML performance itself. Consequently, quite a few technical features that relate to the ML do seem to find their role in the analysis and functioning of the CLuP as well. 

To be able to fully understand the underlying connection we first provided a brief Random Duality Theory (RDT) based technical analysis of the ML and then switched to the corresponding one related to the CLuP. While there are many elements of the analysis that are of great interest, we will here single out one that might be among the most important ones. In particular, there seems to be ceratin (SNR) regimes where the ML performance (or its a very close approximation) might be easier to obtain compared to how difficult it is to obatin the similar one in other regimes. We provided a theoretical characterization of these regimes as well as a discussion how they may relate to the CLuP's performance.

We also discussed various ways as to how the CLuP's performance can be reinforced in the hard regimes as well. Finally, we provided a solid set of results obtained through numerical experiments and observed that they are in a very strong agreement with what the theory predicts.

Since this is the introductory paper on this subject we limit ourselves only to the simplest possible structure of the algorithm. However, we did hint on multiple occasions that we have already explored quite a few other possibilities for building further. All these we will present in great details in a collection of companion papers.

\begin{singlespace}
\bibliographystyle{plain}
\bibliography{clupintRefs}
\end{singlespace}

\end{document}